\documentclass[12pt]{article}

\usepackage[margin=1in]{geometry}
\usepackage{url}
\usepackage{amssymb, amsmath, mathtools, amsfonts, amsthm, dsfont}
\usepackage{float, xcolor, xspace, graphicx, subfigure, appendix, enumitem, booktabs, multirow}
\usepackage{algorithm, algpseudocode}
\usepackage[normalem]{ulem}
\usepackage{array, longtable, setspace}
\usepackage[authoryear,round]{natbib}
\usepackage[bookmarks, colorlinks=true, plainpages=false,
            citecolor=red,
            linkcolor=blue,
            anchorcolor=red,
            urlcolor=blue]{hyperref}
\usepackage{cleveref}
\hypersetup{
    pdftitle={Screening AI-Equipped Applicants},
    pdfauthor={Pengyu Qian and Sophie H. Yu}
}

\makeatletter
\algrenewcommand\alglinenumber[1]{\footnotesize #1. }
\algrenewcommand\algorithmiccomment[1]{\hfill\textcolor{blue}{\(\triangleright\) #1}}
\makeatother
\floatname{algorithm}{Policy}

\usepackage{prettyref}
\newrefformat{eq}{(\ref{#1})}
\newrefformat{sec}{Section~\ref{#1}}
\newrefformat{alg}{Algorithm~\ref{#1}}
\newrefformat{fig}{Figure~\ref{#1}}
\newrefformat{tab}{Table~\ref{#1}}
\newrefformat{rmk}{Remark~\ref{#1}}
\newrefformat{clm}{Claim~\ref{#1}}
\newrefformat{claim}{Claim~\ref{#1}}
\newrefformat{def}{Definition~\ref{#1}}
\newrefformat{assump}{Assumption~\ref{#1}}
\newrefformat{ass}{Assumption~\ref{#1}}
\newrefformat{cor}{Corollary~\ref{#1}}
\newrefformat{lmm}{Lemma~\ref{#1}}
\newrefformat{lem}{Lemma~\ref{#1}}
\newrefformat{prop}{Proposition~\ref{#1}}
\newrefformat{app}{Appendix~\ref{#1}}
\newrefformat{thm}{Theorem~\ref{#1}}

\theoremstyle{plain}
\newtheorem{theorem}{Theorem}
\newtheorem{lemma}{Lemma}
\newtheorem{proposition}{Proposition}
\newtheorem{corollary}{Corollary}

\theoremstyle{definition}
\newtheorem{definition}{Definition}

\newtheorem{remark}{Remark}

\newtheorem{assumption}{Assumption}

\providecommand{\E}{\mathbb E}
\providecommand{\indc}[1]{\mathbf1\{#1\}}

\definecolor{RevisionDarkGreen}{RGB}{0,90,45}

\DeclareRobustCommand{\revredel}[1]{{\color{red}%
\hypersetup{citecolor=red,linkcolor=red,urlcolor=red}\sout{#1}}}
\colorlet{SORevisionBlue}{black}

\usepackage{tikz}
\usetikzlibrary{arrows.meta}

\begin{document}

\title{Screening AI-equipped applicants 
}
\author{
    Pengyu Qian\thanks{
        Questrom School of Business, Boston University.  Email address: \texttt{pqian@bu.edu}.
    }
    \and
    Sophie H.\ Yu\thanks{
        The Wharton School, University of Pennsylvania.  Email address: \texttt{hysophie@wharton.upenn.edu}.
    }
}
\date{}

\maketitle

\begin{abstract}
Generative AI can lower the cost of producing polished application materials without improving applicants' abilities to do the job.
We study an employer who observes qualifications and work history, summarized by observable quality, but must infer other job-relevant attributes from reports that applicants can manipulate at costs that depend on privately known manipulation ability.
The employer commits to a hiring rule under a capacity constraint, and we evaluate rules by match value net of applicants' manipulation
costs.
We show that, as manipulation ability becomes increasingly dispersed among applicants with the same observable quality, socially optimal hiring converges to selection on observable quality alone, a phenomenon we call \emph{strategic feature collapse}.
To capture the remaining welfare gain from reports, we construct a simple dual-track mechanism that directly admits applicants with high observable quality and offers some otherwise excluded applicants a report-based route to admission.
Despite the complexity of multidimensional screening, this mechanism is asymptotically optimal under additional regularity, matching the unrestricted optimum's welfare gain from reports through second order.
\end{abstract}

\onehalfspacing

\setcounter{tocdepth}{2}
\tableofcontents

\section{Introduction}\label{sec:introduction}

Generative AI can make job applicants look better on paper without improving their ability to do the job.
Although AI can also help applicants communicate their strengths more clearly, it can make application materials less informative about the productive qualities employers seek to assess. 
Recent studies document declining employer valuation of customized applications and weaker associations between cover-letter tailoring and hiring outcomes \citep{GaldinSilbert2025,cui2025signaling}.
We study how employer screening should respond as differences in applicants' ability to produce favorable application materials at low cost grow.

Employers observe some applicant characteristics, such as credentials and work history, but must infer other productive qualities. To do so, they also rely on materials applicants produce, such as cover letters and take-home assignments, which we refer to as \emph{reports}.
What these reports reveal is heavily influenced by screening requirements employers impose.
Applicants respond to these requirements, so a favorable report may reflect either productive quality or an advantage in producing such reports at low cost, which we call \emph{manipulation ability}.
Ignoring reports can overlook productive qualities that observable qualifications, such as credentials and work history, do not capture.
Requiring reports, however, can impose costs on applicants even when they contribute little to selection.
The question is when reports are still worth their costs and which applicants should be asked to provide them.

We study an employer who hires a fixed fraction of applicants.
Each applicant has observable quality, a scalar summary of observable qualifications, and latent quality, which captures productive attributes not directly observed by the employer.
Applicants submit reports about latent quality and can inflate their reports at a cost.
This cost depends on {their privately known manipulation ability}, which does not affect the value of the match between the applicant and the employer.
The employer commits to a hiring rule that assigns hiring probabilities based on observable quality and reports, and applicants then choose their reports in response.
There are no monetary transfers.
We evaluate hiring rules by social welfare: the value of the resulting matches less applicants' manipulation costs.
We later consider a private employer who maximizes match value alone, without accounting for applicants' manipulation costs.

Because applicants privately know both their latent quality and manipulation ability, optimal screening is a multidimensional mechanism design problem.
In related settings, optimal mechanisms can require complex menus of randomized options \citep{DaskalakisEtAl2017}.
The key insight of our paper is that growing dispersion in manipulation ability yields a sharp answer.
Socially optimal hiring converges to selection on observable quality alone, a phenomenon we call \emph{strategic feature collapse}, even without restricting the class of screening rules.
Reports can nevertheless improve welfare when dispersion is large.
We construct a simple \emph{dual-track} mechanism that {directly admits applicants with sufficiently high observable quality and offers some otherwise excluded applicants an optional report-based route to admission.
Under additional regularity, this mechanism captures the optimal welfare gain from reports to leading order.}

{To understand why dispersion matters, consider how applicants respond to a report requirement.}
{Applicants with higher latent quality need to inflate their reports less, while those with greater manipulation ability can do so at lower cost.}
We assume that, among applicants with the same observable quality, manipulation ability is independent of latent quality.\footnote{This assumption is plausible in many settings where the skills used to improve application materials differ from those needed for the job. For example, among surgeons with comparable qualifications and experience, proficiency in using AI to polish a resume need not reflect surgical ability.} 
Thus, manipulation ability helps applicants produce stronger reports {without itself signaling latent quality.} 
As manipulation ability becomes more dispersed while the distribution of latent quality remains fixed, it becomes increasingly difficult for screening rules to distinguish productive quality from inexpensive manipulation.

Our first main result shows that this tension constrains every feasible screening mechanism.
Under a general condition requiring manipulation ability to become sufficiently dispersed among applicants with the same observable quality, the additional welfare obtainable from reports vanishes.
Moreover, any asymptotically socially optimal screening mechanisms must select essentially the same applicants as hiring based on observable quality alone.
We establish these results by bounding attainable welfare across all screening rules, including mechanisms that offer randomized hiring probabilities.

Strategic feature collapse shows that the welfare gain from reports vanishes in the large-dispersion limit, but leaves open how much value reports can still generate when dispersion is finite.
We call the optimal welfare gain from using reports relative to hiring based on observable quality alone \emph{residual welfare}.
To quantify this gain, we consider a sequence of markets indexed by $k$, in which manipulation ability follows a location--scale family with a dispersion scale $s_k$ tending to infinity.
A mechanism sequence is \emph{second-order optimal} if its welfare loss relative to the unrestricted optimum is $o(1/s_k)$.

Our second main result shows that, under additional regularity, unrestricted optimal screening and the dual-track mechanism both achieve a welfare gain over observable-only hiring of $D/s_k+o(1/s_k)$, where $D\geq0$ is an instance-dependent constant.
Thus, despite its simple structure, the dual-track mechanism is second-order optimal.

The dual-track mechanism consists of a direct-admission track and a report-based screening track. 
Applicants with sufficiently high observable quality are admitted directly, without a report requirement.
Some applicants below the cutoff for observable-only hiring can instead qualify through the screening track.
Within this track, applicants choose from a common menu of hiring probabilities, with higher probabilities requiring higher reports. The report requirements are adjusted for observable quality.
The two tracks jointly fill capacity, where reports provide an opportunity for otherwise-excluded applicants rather than a requirement imposed on everyone.
As $k$ grows, the screening track fills a vanishing share of positions, of order at most $\frac{1}{s_k}$, 
yet captures all the welfare gains at leading order.

To understand why the screening track can remain informative even when manipulation ability becomes highly dispersed, consider applicants with the same observable quality whose manipulation abilities lie near the upper bound of their group.
This upper bound limits how high a report an applicant of a given latent quality is willing to produce.
An appropriately chosen report requirement, adjusted for observable quality, can therefore deter applicants with low latent quality while remaining attainable for higher latent-quality applicants whose manipulation abilities are near the upper bound.
As dispersion grows, the share of applicants who can qualify through this route shrinks, but those admitted can be better matches than the applicants displaced from direct admission.
The mass of applicants near this boundary determines how much selection the screening track can generate.
Its welfare contribution then depends on the additional match value of the applicants it identifies net of their reporting costs.

Finally, screening rules that generate the same hires can impose different costs on applicants.
An employer concerned only with match value is indifferent between such rules, even if one requires costly reports and another does not.
We show that, as manipulation ability becomes increasingly dispersed, rules approaching the employer-optimal screening select almost the same applicants as socially optimal screening, while their welfare loss equals applicants' manipulation costs up to a vanishing term.
This loss can approach the entire value of admission to applicants. When the common improvement in manipulation ability outpaces its dispersion, an exactly employer-optimal rule can hire the same applicants as observable-only ranking while imposing manipulation costs close to that value. In this example, a costless implementation is also privately optimal.
Thus, even when reports cease to affect who is hired, removing the requirement to produce them can still improve social welfare.

\subsection{Related Literature}\label{subsec:introduction-literature}

Classical signaling theory links the informativeness of an action to differences in its cost across worker types \citep{Spence1973}. 
AI assistance can alter this relationship by changing how applicants produce signals and how employers interpret them, and recent evidence suggests that some conventional application signals are losing value.
\citet{GaldinSilbert2025} document declining employer valuation of customized applications.
\citet{AbbasNejadEtAl2025} and \citet{cui2025signaling}
provide evidence that AI assistance can make cover letters less informative about applicants' qualities. 
Yet the effects of AI writing assistance are not uniformly negative: \citet{WilesEtAl2025} find that nongenerative writing assistance increases hiring by improving the clarity of resumes, and \citet{cowgill2026does} show that AI can improve or impair evaluation depending on how it changes the signals of applicants with different expertise.
Together, these studies suggest that better applications need not lead to better selection, as assistance may weaken the connection between application materials and applicants' productive qualities. We study how employers should optimally design screening rules in response to growing heterogeneity in nonproductive manipulation ability.

The closest work on AI and hiring is the concurrent paper by \citet{AshlagiEtAl2026}, who study costly evaluation by a Bayesian firm using experience and noisy application materials. 
As AI makes these materials less informative, the firm places greater weight on experience and may reject inexperienced applicants who would be good matches. 
They show that a relatively inexpensive intermediate assessment can restore screening opportunities for these applicants.
Our model differs both in the source of information loss and in the employer's response. 
Rather than taking the informativeness of application materials as exogenously given as in \cite{AshlagiEtAl2026}, we allow applicants to choose reports strategically in response to a committed hiring rule.
We show how increasing dispersion in manipulation ability leads to information loss of every feasible screening rule, and how a dual-track mechanism can recover the remaining value of reports without requiring an additional assessment.
Because applicants incur reporting costs, we also distinguish the employer's private objective from social welfare.

Our strategic feature collapse result is closely related to research on decision rules that use manipulable data. 
\citet{FrankelKartik2019} show how heterogeneity in gaming ability muddles the information conveyed by agents' actions. 
\citet{FrankelKartik2022} show that a designer who commits to a linear rule should underuse manipulable data. The optimal weight on the data vanishes as gaming ability becomes more dispersed relative to the natural action. \citet{Ball2025} studies linear scoring with privately known distortion costs and characterizes which features should be underweighted or overweighted. 
Both papers optimize within the class of linear rules. Our strategic feature collapse result instead bounds the value of every feasible screening rule and reporting equilibrium, including randomized menus, subject to a capacity constraint.
This extension matters because nonlinear menus can exploit features of the manipulation ability distribution, such as its upper endpoint, even when linear scoring becomes ineffective.

{Allowing general report-based rules also connects our analysis to mechanism design under costly falsification.}
\citet{PerezRichetSkreta2022} show that optimal tests may offer approval to some agents whose true characteristics would not warrant it, thereby reducing incentives to falsify. \citet{PerezRichetSkreta2023} characterize allocation rules without transfers that make truthful reporting optimal and give conditions under which such rules maximize welfare.
Our mechanisms need not induce truthful reports: applicants choose strategically among report requirements, and the employer must screen along two privately known dimensions, latent quality and manipulation ability.
We characterize the limits of this self-selection and its residual welfare gains.

The same strategic response lies at the heart of strategic classification, which studies costly changes to classifier inputs \citep{HardtEtAl2016}, and performative prediction, which studies feedback from deployed rules to the data distribution \citep{perdomo2020performative}.
\citet{CohenEtAl2025} consider hiring with heterogeneous access to LLM manipulation, while \citet{BravermanGarg2020} demonstrate the role of randomization when manipulation costs enter efficiency.
\citet{LiuGargBorgs2022} study capacity-constrained admission by rank when effort is costly, and \citet{MilliEtAl2019} emphasize the social burden of strategic manipulation. 
Our welfare criterion likewise counts applicants' costs. We show that privately optimal screening can impose substantial reporting costs even when reports no longer improve selection.
More generally, we establish that no screening rule, however nonlinear or randomized, can sustain a nonvanishing gain in match value from reports as manipulation ability becomes increasingly dispersed.

The cost of eliciting information also connects our analysis to allocation without money. Ordeals can improve targeting when they burden unintended recipients more \citep{NicholsZeckhauser1982}, while the desirability of screening rather than lotteries or priority rules depends on the distribution of values and the costs of signaling 
\citep{HartlineRoughgarden2008,Condorelli2012,ChakravartyKaplan2013}. 
\citet{YangDworczakAkbarpour2023} compare screening devices by their targeting effectiveness and the rents they leave to agents. In a model with transfers, \citet{Yang2025} gives conditions under which a principal can dispense with socially costly screening instruments. 
Related targeting problems arise when agents' willingness to pay or ability does not coincide with the designer's valuation \citep{CheGaleKim2013,Condorelli2013}.
In our setting, manipulation ability similarly affects the cost of qualifying without directly affecting match value. We show that the resulting private-social welfare gap can approach applicants' entire benefit from admission.

Our dual-track mechanism balances the value and cost of reporting by directly admitting applicants with high observable quality and offering otherwise excluded applicants an optional report-based route.
Direct admission relates to tagging based on observable characteristics \citep{Akerlof1978}, while the alternative screening route resembles costly verification \citep{BenPorathEtAl2014,KattwinkelKnoepfle2023} and test-optional admissions \citep{DesseinEtAl2025,GargLiMonachou2026}.
Unlike verification, however, our mechanism introduces no separate assessment; it induces self-selection through existing report requirements.
We also characterize when a single report hurdle suffices and when intermediate admission probabilities can improve welfare, a distinction related to the role of cost curvature in optimal prize design \citep{MoldovanuSela2001}.

Finally, our second-order analysis contributes to research on when simple mechanisms perform as well as unrestricted multidimensional screening.
Related work studies the optimality of pure or nested bundling \citep{haghpanah2021pure,yang2025nested}.
A particularly close methodological precedent is \citet{FrickIijimaIshii2026}, who show that when a seller's information about a buyer becomes increasingly precise, selling the grand bundle achieves the optimal convergence rate to first-best revenue, whereas separate sales converge more slowly.
Our informational limit runs in the opposite direction: as manipulation ability becomes increasingly dispersed, reports lose their value and observable-only hiring becomes optimal to first order.
We show that a simple dual-track mechanism nonetheless matches the unrestricted optimum's leading residual welfare gain.

\section{The Model}\label{sec:model}

\subsection{Applicants and match values}

We consider a hiring market. There is a unit mass of applicants, and the employer fills a measure $\rho\in(0,1)$ of positions.
Each applicant is characterized by observable quality $Z\in\mathbb{R}$, latent quality $X\in\mathbb{R}$, and manipulation ability $\Theta\in\mathbb{R}_{++}$.
The employer observes $Z$, while $X$ and $\Theta$ are private information.

If matched with an applicant of type $(z,x,\theta)$, the employer obtains match value $v(z,x)$.
Manipulation ability is nonproductive: it affects the cost of producing a report but does not directly affect match value.

\begin{assumption}[Type distribution]\label{ass:conditional-independence}
The type vector $(Z,X,\Theta)$ has a joint density that factorizes as $$f_{Z,X,\Theta}\left(z,x,\theta\right)=f_Z(z)g\left(x\mid z\right)f\left(\theta\mid z\right),\qquad \left(z,x,\theta\right)\in\mathbb R^2\times\left(0,\infty\right),$$ where $g$ and $f$ are jointly measurable conditional densities of latent quality and manipulation ability given observable quality.
The density $f_Z$ is continuous and strictly positive on $\mathbb R$.
Moreover, $\E[X^2+v(Z,X)^2]<\infty$.
\end{assumption}

The factorization in Assumption~\ref{ass:conditional-independence} imposes $X\perp\Theta\mid Z$.
Thus, after conditioning on the information observable to the employer, manipulation ability carries no additional information about latent quality.
Observable and latent quality may nevertheless be correlated. The assumption rules out manipulation ability that varies systematically with latent quality within an observable group. Evidence that AI assistance helps less able writers and lower-quality applicants most \citep{NoyZhang2023,AbbasNejadEtAl2025} motivates considering such dependence. These studies, however, do not directly test our conditional independence assumption.
We refer to applicants with the same observable quality $z$ as an observable group.
Conditional distributional statements below are understood to hold for almost every $z$ under the distribution of $Z$.

We impose no particular functional form on $v$, and define the expected match value based only on observable quality by $m(z)\triangleq \E[v(z,X)\mid Z=z]$.

\begin{assumption}[Monotone quality]\label{ass:type-monotonicity}
The match value $v(z,x)$ is weakly increasing in $x$ for every $z$.
The conditional expectation
$m(z)=\E[v(z,X)\mid Z=z]$ admits a continuous,
strictly increasing version on $\mathbb R$, which we use throughout.
\end{assumption}

Economically, Assumption~\ref{ass:type-monotonicity} says that higher latent quality is never harmful to the employer, while higher observable quality corresponds to a higher expected match value before screening.
Absent additional screening, the employer therefore ranks applicants by observable quality.
Continuity ensures that small changes in observable quality produce small changes in expected match value.

\subsection{Manipulation}

After the employer observes $z$, the applicant submits a report $r\in\mathbb{R}$ about latent quality $x$.
Manipulation cost is one-sided: reporting below latent quality is costless, while reporting above latent quality is costly.
An applicant with latent quality $x$ and manipulation ability $\theta$ who submits report $r$ incurs cost
\begin{equation}
    c\left(\frac{(r-x)_+}{\theta}\right). \label{eq:manipulation-cost}
\end{equation}
Higher $\theta$ therefore makes any given upward manipulation less costly.

\begin{assumption}[Manipulation cost]\label{ass:manipulation-cost}
The function $c:[0,\infty)\to[0,\infty)$ is continuous,
strictly increasing, and strictly convex. It satisfies
$c(0)=0$ and $\lim_{u\to\infty}c(u)=\infty$.
\end{assumption}

There are no monetary transfers.
Allocation means admission to a position.
An applicant obtains benefit one from admission and zero otherwise.
Thus, if submitting report $r$ results in allocation probability $q\in[0,1]$, the applicant's expected payoff is \( q-c\left(\frac{(r-x)_+}{\theta}\right).\)

\subsection{The screening game}\label{subsec:screening-game}

The interaction between the employer and the applicants is a two-stage commitment game.
First, the employer publicly commits to an allocation rule contingent on observable quality and the report.
Second, after observing the rule, each applicant decides which report to submit.

A screening rule is a Borel measurable function
$$
Q:\mathbb{R}\times\mathbb{R}\to\left[0,1\right],
$$
where $Q(z,r)$ is the admission probability for an applicant with observable quality $z$ who submits report $r$.
For a given screening rule $Q$, define the applicant's payoff from a report by
\begin{equation}\label{eq:agentproblem}
u_Q\left(r;z,x,\theta\right)
\triangleq
Q\left(z,r\right)
-
c\left(\frac{(r-x)_+}{\theta}\right),
\qquad
r\in\mathbb{R}.
\end{equation}

\begin{definition}[Reporting equilibrium]\label{def:reporting-equilibrium}
Fix a screening rule $Q$.
A reporting equilibrium under $Q$ is a measurable function
$
R_Q:\mathbb{R}^2\times (0,\infty)\to\mathbb{R}
$
such that
$$
R_Q\left(z,x,\theta\right)
\in
\arg\max_{r\in\mathbb{R}}
u_Q\left(r;z,x,\theta\right)
$$
for every type $(z,x,\theta)$.
A screening rule is admissible if a reporting equilibrium exists.
\end{definition}

Given an admissible rule $Q$ and a reporting equilibrium $R_Q$, define the type-dependent outcomes by $$A_Q\left(z,x,\theta\right)\triangleq Q\left(z,R_Q\left(z,x,\theta\right)\right),\qquad C_Q\left(z,x,\theta\right)\triangleq c\left(\frac{(R_Q\left(z,x,\theta\right)-x)_+}{\theta}\right).$$
Write $A_Q\triangleq A_Q(Z,X,\Theta)$ and $C_Q\triangleq C_Q(Z,X,\Theta)$ for the corresponding random outcomes.

Applicants do not interact after the employer commits to $Q$, so a reporting equilibrium consists of their measurable best responses.
The employer anticipates these responses when choosing the screening rule.

A pair $(Q,R_Q)$ is feasible if $Q$ is admissible, $R_Q$ is a reporting equilibrium under $Q$, and the induced allocation satisfies
\begin{equation}\label{eq:capacity}
\E[A_Q]=\rho.
\end{equation}

The optimal normalized social welfare is
\begin{equation}\label{eq:value}
W^\star
\triangleq
\sup_{\left(Q,R_Q\right)}
\left\{
\E\left[v\left(Z,X\right)A_Q-C_Q\right]:
\left(Q,R_Q\right)\text{ is feasible}
\right\}.
\end{equation}
Applicants' allocation benefits contribute the constant $\E[A_Q]=\rho$, which is omitted from normalized social welfare.
A private employer that does not internalize applicants' manipulation costs instead maximizes \begin{equation}\label{eq:employer-utility}U\triangleq \sup_{\left(Q,R_Q\right)}\left\{\E\left[v\left(Z,X\right)A_Q\right]:\left(Q,R_Q\right)\text{ is feasible}\right\}.\end{equation}

For each observable quality $z$, the screening rule $Q$ offers the menu
\[
\left\{(Q(z,r),r):r\in\mathbb{R}\right\}.
\]
Each contract $(q,r)$ gives admission probability $q$ upon submission
of report $r$.
Applicants choose among these contracts to maximize their admission
benefit net of manipulation cost.

\subsection{Technological change and the asymptotic analysis}\label{sec:asymptotic-regime}

Improvements in manipulation technology change what an employer can learn from applicants' reports.
For a given distribution of manipulation ability, the employer chooses a screening rule, and applicants select among its contracts using both their latent quality and their manipulation ability.
Latent quality determines the point above which reporting becomes costly, while manipulation ability determines the cost of raising a report above that point.
These two dimensions of private information govern self-selection across the menu, which can contain a continuum of contracts.
This is a multidimensional screening problem in which allocation decisions and incentive constraints must be determined jointly \citep{RochetChone1998}.
We study how the optimal screening rule changes as the distribution of manipulation ability changes.

Consider a sequence of hiring markets indexed by $k$.
The joint distribution of $(Z,X)$, the match-value function $v$, the manipulation cost function $c$, and capacity $\rho$ remain fixed.
Let $\Theta_k>0$ denote manipulation ability in market $k$. Its conditional distribution given $Z$ may vary with $k$.
Assumption~\ref{ass:conditional-independence} is maintained in every market with $\Theta$ replaced by $\Theta_k$.
The comparison is across market distributions. It does not require identifying applicants across markets.

In market $k$, an applicant with latent quality $x$ and manipulation ability $\theta_k$ who submits report $r$ incurs cost $$c\left(\frac{(r-x)_+}{\theta_k}\right).$$
Changes in the distribution of $\Theta_k$ therefore affect both how far applicants can manipulate and how much this capability differs across applicants with the same observable quality.
The distinction matters for screening: a common improvement makes a given report requirement easier to meet, while greater dispersion changes the information that passing the requirement conveys about latent quality.

The employer chooses a screening rule $Q_k$ in each market and anticipates the reporting responses induced by that rule.
Report requirements, allocation probabilities, and menu complexity can all change with the technology.
Let $W_k^\star$ and $U_k$ denote, respectively, the optimal normalized social welfare and the employer's optimal expected match value, as defined in \eqref{eq:value} and \eqref{eq:employer-utility} for market $k$.

The first-order analysis identifies when observable quality alone determines the limiting optimal allocation, a phenomenon we call {strategic feature collapse}.
Under additional regularity, the second-order analysis characterizes the optimal welfare coefficient at scale $1/s_k$ and constructs a sequence of menus attaining it.
Together, these results characterize the limiting allocation and the remaining role of report-based screening.

\section{First-Order Behavior: Strategic Feature Collapse}\label{sec:collapse}

A high report can reflect either high latent quality or high manipulation ability.
The employer values the former but must screen applicants using a report influenced by both.
We show that, as manipulation ability becomes sufficiently dispersed, the welfare gain from report-based screening vanishes.
Ranking applicants by observable quality is asymptotically socially optimal, and the allocation induced by every asymptotically socially optimal sequence approaches this ranking.
We call this phenomenon strategic feature collapse.

\subsection{Dispersion in manipulation ability}

Consider the sequence of markets defined in
Section~\ref{sec:asymptotic-regime}.
The relevant dispersion is the variation in manipulation ability
that remains after conditioning on the employer's observable
information.
For $d>0$, define the conditional L\'evy concentration function by
\begin{equation}\label{eq:collapse-concentration}
\omega_k(d\mid z)
\triangleq
\sup_{a\in\mathbb R}
\Pr\{a\leq\Theta_k\leq a+d\mid Z=z\}.
\end{equation}
This is the largest share of applicants with observable quality $z$
whose manipulation abilities lie within an interval of width $d$.
The interval may be chosen separately for each group and market.

\begin{assumption}[Conditional anti-concentration]
\label{ass:levy-anticoncentration}
For every fixed $d>0$,
\begin{equation}\label{eq:collapse-levy}
\E[\omega_k(d\mid Z)]
\longrightarrow 0
\qquad\text{as }k\to\infty.
\end{equation}
\end{assumption}

The expectation in \eqref{eq:collapse-levy} averages the within-group concentration measures according to the population distribution of $Z$.
The assumption requires this average to vanish for every fixed interval width.
Equivalently, for any fixed $d>0$ and $\varepsilon>0$, the share of applicants belonging to groups with $\omega_k(d\mid z)\geq\varepsilon$ must converge to zero.
Thus, groups in which a fixed-width interval contains a substantial fraction of applicants must account for a vanishing share of the population.

For example, if $\Theta_k\mid Z=z$ is uniform on $[1,k+1]$ in every observable group, then $\omega_k(d\mid z)=\min\{1,d/k\}$, and the assumption holds.
If instead $\Theta_k\mid Z=z$ is uniform on $[k,k+1]$, then $\omega_k(1\mid z)=1$, and the assumption fails.
In the latter case, manipulation ability increases for everyone, but the within-group distribution retains its fixed width.
A common increase in manipulation ability alone therefore does not imply conditional anti-concentration.
\subsection{Strategic feature collapse}

We call admission in descending order of $Z$ observable-only ranking.
Let $F_Z$ be the distribution function of $Z$ and set $z_0\triangleq F_Z^{-1}(1-\rho)$ and $\lambda\triangleq m(z_0)$.
Since $f_Z$ is strictly positive on $\mathbb R$, $z_0$ is the unique observable-quality cutoff satisfying $\Pr\{Z>z_0\}=\rho$.
Strict increase of $m$ gives $\Pr\{m(Z)>\lambda\}=\rho$.
Thus $z_0$ is the cutoff in observable quality, while $\lambda$ is the corresponding cutoff in expected match value.

The benchmark allocation and its value are
\begin{equation}\label{eq:collapse-benchmark}
\begin{aligned}
a^0(z)&\triangleq \mathbf1\{m(z)>\lambda\},\\
V^0&\triangleq \E[m(Z)a^0(Z)].
\end{aligned}
\end{equation}
The rule admits applicants above the observable-quality cutoff.
The screening rule $Q_k(z,r)\triangleq a^0(z)$ implements it at zero manipulation cost in every market.

For a feasible pair $(Q_k,R_{Q_k})$, write $A_k(z,x,\theta)\triangleq A_{Q_k}(z,x,\theta)$ and $C_k(z,x,\theta)\triangleq C_{Q_k}(z,x,\theta)$ for its type-dependent outcomes.
Let $A_k\triangleq A_k(Z,X,\Theta_k)$ and $C_k\triangleq C_k(Z,X,\Theta_k)$, and denote its private payoff and social welfare by
\[
\pi_k\triangleq \E[v\left(Z,X\right)A_k],
\qquad
W_k\triangleq \pi_k-\E[C_k].
\]
Thus $W_k$ denotes the welfare of the specified feasible pair, while $W_k^\star$ is the supremum over all feasible pairs.
For any feasible pair, define its loss from departing from observable-only ranking by
\[
L_k\triangleq \E\big[|m(Z)-\lambda|\,|A_k-a^0(Z)|\big].
\]
A feasible sequence is \emph{asymptotically socially optimal} if
\[
W_k^\star-W_k\longrightarrow0
\qquad\text{as }k\to\infty.
\]
The difference $W_k^\star-V^0$, which we call residual welfare, is the welfare gain available from using reports in addition to observable quality, holding the information environment and manipulation technology fixed. The corresponding gain for a private employer is $U_k-V^0$.

\begin{theorem}[Strategic feature collapse]\label{thm:strategic-feature-collapse}
Under the assumptions of Section~\ref{sec:model} and
Assumption~\ref{ass:levy-anticoncentration},
\begin{equation}\label{eq:collapse-value}
W_k^\star-V^0\longrightarrow0,
\qquad
U_k-V^0\longrightarrow0
\qquad\text{as }k\to\infty.
\end{equation}
Every asymptotically socially optimal sequence satisfies
\begin{equation}\label{eq:collapse-allocation}
\begin{aligned}
\E\big[|A_k-a^0(Z)|\big]&\longrightarrow0,\\
\E[C_k]&\longrightarrow0,
\end{aligned}
\qquad\text{as }k\to\infty.
\end{equation}
\end{theorem}

The theorem identifies both who is selected and the cost of selection.
As manipulation ability disperses, the additional match value obtainable from reports vanishes for both social and private objectives.
An asymptotically socially optimal mechanism therefore approaches observable-only ranking and imposes vanishing aggregate manipulation costs.
\citet{FrankelKartik2022} reach the analogous conclusion for linear rules: the optimal weight on manipulable data vanishes as gaming ability becomes more dispersed. Here the value bound is uniform over feasible screening rules and reporting equilibria, including randomized menus.

The following bound quantifies the gains from report-based screening.
It compares two independent applicants with the same observable quality, measuring their latent-quality difference relative to their manipulation-ability difference and weighting the comparison by their difference in match value.

\begin{proposition}[Loss bound for observable-only screening]\label{prop:collapse-rate}
Under the assumptions of Theorem~\ref{thm:strategic-feature-collapse}, let $(X',\Theta_k')$ be an independent copy of $(X,\Theta_k)$ conditional on $Z$, and write
\[
\Delta X\triangleq |X-X'|,
\qquad
\Delta\Theta_k\triangleq |\Theta_k-\Theta_k'|,
\qquad
\Delta v\triangleq |v\left(Z,X\right)-v\left(Z,X'\right)|.
\]
For a fixed latent-quality gap $d\geq0$, define
\begin{equation}\label{eq:collapse-pairwise-kernel}
\eta_k\left(d\mid z\right)
\triangleq \E\left[\min\left\{1,c\left(\frac{d}{\Delta\Theta_k}\right)\right\}
\middle|Z=z\right],
\end{equation}
where the minimum is defined to equal one when $\Delta\Theta_k=0$.
Define
\begin{equation}\label{eq:collapse-modulus}
\varepsilon_k
\triangleq \frac12\E\left[\Delta v\sqrt{\eta_k\left(\Delta X\mid Z\right)}\right].
\end{equation}
Then, for every $k$,
\begin{equation}\label{eq:collapse-rate}
V^0\leq W_k^\star\leq U_k\leq V^0+\varepsilon_k.
\end{equation}
Moreover, $\varepsilon_k\to0$ as $k\to\infty$.
\end{proposition}

For fixed $z,x,x'$, the factor $\sqrt{\eta_k(|x-x'|\mid z)}$ bounds the difference between the two latent-quality types' admission probabilities after averaging each over $\Theta_k\mid Z=z$.
The factor $\Delta v$ measures their difference in match value.
Their weighted average bounds the additional match value obtainable by using reports to distinguish applicants with the same observable quality.

The convergence of $\varepsilon_k$ follows from the different behavior of manipulation-ability gaps and latent-quality gaps.
Conditional anti-concentration makes manipulation-ability gaps diverge in probability as $k\to\infty$, averaging over observable groups, while latent-quality gaps have a fixed distribution. The bounded cost comparison therefore vanishes on average as $k\to\infty$.
Integrability of match values controls the contribution of large $\Delta v$, so the weighted bound also vanishes as $k\to\infty$.
The final part of the proposition's proof gives the details.

\begin{proof}[Proof of Theorem~\ref{thm:strategic-feature-collapse}]
Proposition~\ref{prop:collapse-rate} gives
\[
0\leq W_k^\star-V^0\leq U_k-V^0\leq\varepsilon_k,
\qquad
\varepsilon_k\to0\quad\text{as }k\to\infty.
\]
The squeeze theorem therefore proves both limits in
\eqref{eq:collapse-value}.

Now consider an asymptotically socially optimal sequence.
The welfare decomposition \eqref{eq:collapse-decomposition} and the information bound \eqref{eq:collapse-information} imply
\[
\begin{aligned}
0\leq L_k+\E[C_k]
&=\E[(v\left(Z,X\right)-m(Z))A_k]-(W_k-V^0)\\
&\leq\varepsilon_k+(W_k^\star-W_k),
\end{aligned}
\]
where we used $W_k^\star\geq V^0$. The right-hand side tends to zero as $k\to\infty$. Since both terms on the left are nonnegative,
\begin{equation}\label{eq:collapse-weighted-allocation}
L_k\longrightarrow0,
\qquad
\E[C_k]\longrightarrow0
\qquad\text{as }k\to\infty.
\end{equation}

To obtain unweighted allocation convergence, split applicants according to the distance of $m(Z)$ from $\lambda$. For every $\delta>0$,
\[
\E[|A_k-a^0(Z)|]
\leq\Pr\{|m(Z)-\lambda|\leq\delta\}+\frac{L_k}{\delta}.
\]
The first term bounds disagreement near the cutoff, while the second bounds disagreement elsewhere. For fixed $\delta$, $L_k/\delta\to0$ as $k\to\infty$. Since $Z$ is atomless and $m$ is strictly increasing, $m(Z)$ has no atom at $\lambda$, so
\[
\Pr\{|m(Z)-\lambda|\leq\delta\}\longrightarrow0
\qquad\text{as }\delta \to 0.
\]
Taking $k\to\infty$ first and then $\delta\to0$ proves
\[
\E[|A_k-a^0(Z)|]\longrightarrow0
\qquad\text{as }k\to\infty,
\]
completing the proof.
\end{proof}

\begin{remark}[Private employers]
Call a feasible sequence \emph{asymptotically privately optimal} if $U_k-\pi_k\to0$ as $k\to\infty$.
For such a sequence, \eqref{eq:collapse-information} and \eqref{eq:collapse-decomposition}, with $W_k=\pi_k-\E[C_k]$, give
\[
0\leq L_k\leq\varepsilon_k+(U_k-\pi_k).
\]
Since both terms on the right tend to zero as $k\to\infty$, it follows that
\[
L_k\longrightarrow0
\qquad\text{as }k\to\infty.
\]
The same cutoff argument therefore gives
\[
\E[|A_k-a^0(Z)|]\longrightarrow0
\qquad\text{as }k\to\infty.
\]
This argument does not imply that applicants' manipulation costs vanish as $k\to\infty$.
\end{remark}

\subsubsection{Proof of Proposition~\ref{prop:collapse-rate}}
\label{sec:proof-collapse-rate}

Fix any feasible screening rule and reporting equilibrium in market $k$.
Since reporting $r=X$ is costless, optimality implies $0\leq C_k\leq1$.
Together with $\E[|v(Z,X)|]<\infty$, this ensures that the payoff and welfare terms below are integrable.

We first compare the allocations of two applicants facing the same rule.

\begin{lemma}[Cross-type allocation comparison]\label{lem:collapse-comparison}
Fix a screening rule and consider applicants of types $(z,x_1,\theta_1)$ and $(z,x_2,\theta_2)$, where $\theta_2>\theta_1>0$. Suppose their optimal reports induce allocation probabilities $q_1$ and $q_2$, respectively. Then
\begin{equation}\label{eq:collapse-comparison}
q_1-q_2
\leq
c\left(\frac{(x_1-x_2)_+}{\theta_2-\theta_1}\right).
\end{equation}
In particular, holding both observable and latent quality fixed, equilibrium allocation is nondecreasing in manipulation ability.
\end{lemma}

Fix $z$ for which the conditional ability distribution is atomless.
This holds for almost every $z$. Define the average allocation at $(z,x)$ by
\[
\bar A_k\left(z,x\right)
\triangleq \E\left[A_k\left(z,x,\Theta_k\right)\mid Z=z\right].
\]
Let $\Theta_k,\Theta_k'$ be independent draws from the conditional ability distribution given $Z=z$. Given $\Theta_k<\Theta_k'$,
Lemma~\ref{lem:collapse-comparison} gives
\begin{equation}\label{eq:collapse-ineq-1}
A_k\left(z,x,\Theta_k\right)-A_k\left(z,x',\Theta_k'\right)
\leq c\left(\frac{(x-x')_+}{\Theta_k'-\Theta_k}\right).
\end{equation}

To average this comparison over manipulation ability, we use the following inequality.

\begin{lemma}[Averaging inequality]\label{lem:collapse-averaging}
For any measurable functions $f,g:(0,1)\to[0,1]$,
\begin{equation}\label{eq:collapse-monotone-pair}
\left(\int_0^1[f(u)-g(u)]\,du\right)_+^2
\leq
2\int_0^1\int_u^1(f(u)-g(v))_+\,dv\,du.
\end{equation}
\end{lemma}

An order-preserving quantile change of variables applies Lemma~\ref{lem:collapse-averaging} to the atomless conditional ability distribution. Combining it with the preceding allocation comparison gives
\[
\begin{aligned}
\left(\bar A_k\left(z,x\right)-\bar A_k\left(z,x'\right)\right)_+^2
&\leq2\E\left[
\min\left\{1,c\left(\frac{|x-x'|}{\Delta\Theta_k}\right)\right\}
\mathbf1\{\Theta_k<\Theta_k'\}\,\middle|\,Z=z\right]\\
&=\eta_k\left(|x-x'|\mid z\right),
\end{aligned}
\]
where the inequality uses \eqref{eq:collapse-ineq-1} and $0\leq A_k\leq1$, and the equality follows by symmetry. Exchanging $x$ and $x'$ gives
\begin{equation}\label{eq:collapse-conditional-comparison}
|\bar A_k\left(z,x\right)-\bar A_k\left(z,x'\right)|
\leq\sqrt{\eta_k\left(|x-x'|\mid z\right)}.
\end{equation}

Conditional independence and symmetrization give
\begin{align}
\E[(v\left(Z,X\right)-m(Z))A_k]
& \overset{(a)}{=}\E[(v\left(Z,X\right)-m(Z))\bar A_k\left(Z,X\right)] \nonumber \\
&  \overset{(b)}{=} \frac12\E\Big[
\left(v\left(Z,X\right)-v\left(Z,X'\right)\right) \nonumber \\
&\qquad\qquad\times
\left(\bar A_k\left(Z,X\right)-\bar A_k\left(Z,X'\right)\right)\Big]\nonumber \\
&  \overset{(c)}{\le}\frac12\E\left[\Delta v\sqrt{\eta_k\left(\Delta X\mid Z\right)}\right]
=\varepsilon_k, \label{eq:collapse-information}
\end{align}
where $(a)$ follows from iterated expectations and conditional independence of $X$ and $\Theta_k$ given $Z$. Equality $(b)$ follows by expanding the product and using that $X$ and $X'$ are independent and identically distributed conditional on $Z$. Inequality $(c)$ holds by \eqref{eq:collapse-conditional-comparison}.

To compare payoffs with observable-only ranking, write
\[
\begin{aligned}
\pi_k-V^0
&=\E[(v\left(Z,X\right)-m(Z))A_k]
  +\E[(m(Z)-\lambda)(A_k-a^0(Z))]\\
&=\E[(v\left(Z,X\right)-m(Z))A_k]-L_k\\
&\leq\varepsilon_k,
\end{aligned}
\]
where the first equality uses $V^0=\E[m(Z)a^0(Z)]$ and $\E[A_k-a^0(Z)]=0$, the second uses $a^0(Z)=\mathbf1\{m(Z)>\lambda\}$ and the definition of $L_k$, and the inequality follows from \eqref{eq:collapse-information} and $L_k\geq0$.
Since $W_k=\pi_k-\E[C_k]$, this also gives
\begin{equation}\label{eq:collapse-decomposition}
W_k-V^0
=\E[(v\left(Z,X\right)-m(Z))A_k]-L_k-\E[C_k].
\end{equation}

Since this bound holds for every feasible screening rule and reporting equilibrium, the optimal private payoff satisfies $U_k\leq V^0+\varepsilon_k$.
Nonnegative costs imply $W_k^\star\leq U_k$, and observable-only ranking attains $V^0$ at zero cost. Hence
\[
V^0\leq W_k^\star\leq U_k\leq V^0+\varepsilon_k,
\]
which proves \eqref{eq:collapse-rate}.

It remains to show that $\varepsilon_k\to0$ as $k\to\infty$.
By the definition of $\omega_k$ in \eqref{eq:collapse-concentration}, for every fixed $d>0$,
\[
\Pr(\Delta\Theta_k\leq d)
\leq\E[\omega_k\left(2d\mid Z\right)],
\]
because, conditional on $Z$ and $\Theta_k'$, the event on the left restricts $\Theta_k$ to an interval of length $2d$.
The right-hand side tends to zero as $k\to\infty$.
Thus $\Delta\Theta_k\xrightarrow{p}\infty$ as $k\to\infty$.
Since the distribution of $\Delta X$ does not depend on $k$, $\frac{\Delta X}{\Delta\Theta_k}\xrightarrow{p}0$ as $k\to\infty$.

For any $M>0$,
\[
\begin{aligned}
2\varepsilon_k
&\leq M\E\left[\sqrt{\eta_k\left(\Delta X\mid Z\right)}\right]
   +\E[\Delta v\,\mathbf1\{\Delta v>M\}]\\
&\leq M\left(
\E\left[\min\left\{1,c\left(
\frac{\Delta X}{\Delta\Theta_k}\right)\right\}\right]
\right)^{1/2}
   +\E[\Delta v\,\mathbf1\{\Delta v>M\}],
\end{aligned}
\]
where the first inequality splits the expectation at $\Delta v=M$ and uses $\eta_k\leq1$, and the second follows from Jensen's inequality and conditional independence.

For fixed $M$, the first term in the last line tends to zero as $k\to\infty$, because $\Delta X/\Delta\Theta_k\xrightarrow{p}0$ as $k\to\infty$, $c$ is continuous at zero with $c(0)=0$, and the integrand is bounded by one.
Consequently,
\[
0\leq\limsup_{k\to\infty}\varepsilon_k
\leq\frac12\E[\Delta v\,\mathbf1\{\Delta v>M\}].
\]
Since $\E[\Delta v]\leq2\E[|v(Z,X)|]<\infty$, the right-hand side tends to zero as $M\to\infty$. Hence $\varepsilon_k\to0$ as $k\to\infty$.

We are left to prove the two lemmas.
\begin{proof}[Proof of Lemma~\ref{lem:collapse-comparison}]
Let $r_i$ be applicant $i$'s chosen report and write $K_i(r)\triangleq c((r-x_i)_+/\theta_i)$.
It suffices to consider $q_1>q_2$. Optimality gives
\begin{equation}\label{eq:collapse-pairwise-ic}
K_1(r_1)-K_1(r_2)
\leq q_1-q_2
\leq K_2(r_1)-K_2(r_2),
\end{equation}
where the second inequality implies $r_1>r_2$, since $K_2$ is nondecreasing.

Let $d\triangleq (x_1-x_2)_+$ and define
\[
r^*\triangleq x_1+\frac{\theta_1d}{\theta_2-\theta_1},
\qquad
K_1(r^*)=K_2(r^*)
=c\left(\frac{d}{\theta_2-\theta_1}\right).
\]
Write $D(r)\triangleq K_2(r)-K_1(r)$.
Since $c$ is convex and nondecreasing, $K_1$ and $K_2$ are convex and hence locally absolutely continuous.
Thus $D$ is locally absolutely continuous.
For $r\leq r^*$, applicant 1's normalized manipulation is weakly smaller, so $D(r)\geq0$. For $r>r^*$, this ordering reverses, and convexity of $c$ together with $\theta_1<\theta_2$ gives $D'(r)<0$ almost everywhere.
Thus $D$ decreases strictly from $D(r^*)=0$ above $r^*$.
If $r_1>r^*$, we would therefore have $D(r_1)<D(r_2)$, contradicting \eqref{eq:collapse-pairwise-ic}, which implies $D(r_1)\geq D(r_2)$. Hence $r_1\leq r^*$, and
\[
q_1-q_2
\leq K_2(r_1)-K_2(r_2)
\leq K_2(r^*)
=c\left(\frac{(x_1-x_2)_+}{\theta_2-\theta_1}\right),
\]
where the first inequality follows from \eqref{eq:collapse-pairwise-ic}, the second uses $r_1\leq r^*$ and the nonnegativity and monotonicity of $K_2$, and the equality follows from the choice of $r^*$.
When $x_1=x_2$, the bound is zero, proving the stated monotonicity.
\end{proof}

\begin{proof}[Proof of Lemma~\ref{lem:collapse-averaging}]
For $t\in[0,1]$, let $S_t\triangleq \{u:g(u)\leq t<f(u)\}$ and let $s(t)$ denote its Lebesgue measure. If $u,v\in S_t$, then $g(v)\leq t<f(u)$, and the set of such pairs with $u<v$ has area $s(t)^2/2$ by symmetry. Integrating over $t$ and applying Cauchy--Schwarz therefore gives
\[
\begin{aligned}
2\int_0^1\int_u^1(f(u)-g(v))_+\,dv\,du
&\geq\int_0^1s(t)^2\,dt
\geq\left(\int_0^1s(t)\,dt\right)^2\\
&=\left(\int_0^1(f(u)-g(u))_+\,du\right)^2
\geq\left(\int_0^1[f(u)-g(u)]\,du\right)_+^2.
\end{aligned}
\]
The equality follows by integrating the definition of $S_t$ in the opposite order. This proves \eqref{eq:collapse-monotone-pair}.
\end{proof}

\subsection{Common improvements and growing dispersion}
\label{subsec:asymptotic-regime}

Technological progress can make manipulation easier for everyone while also widening differences in manipulation ability.
We use a location--scale family to separate these two changes.
Fix a reference joint distribution of $(Z,X,\Theta)$ satisfying
Assumption~\ref{ass:conditional-independence}, and specify manipulation ability in market $k$ by
\begin{equation}\label{eq:scaled-ability}
\left(\Theta_k\mid Z=z\right)
\stackrel{d}{=}\left(\ell_k+s_k\Theta\mid Z=z\right),
\qquad \ell_k\geq0,\quad s_k>0.
\end{equation}
The location parameter $\ell_k$ captures the common improvement, while the scale parameter $s_k$ controls dispersion.
The reference distribution remains fixed as $k$ varies.

The following corollary shows that growing dispersion is sufficient for strategic feature collapse, regardless of the size of the common improvement.

\begin{corollary}[Location--scale growth]
\label{cor:collapse-location-scale}
Suppose Assumptions~\ref{ass:conditional-independence}, \ref{ass:type-monotonicity}, and \ref{ass:manipulation-cost} hold, and manipulation ability follows \eqref{eq:scaled-ability}.
If $s_k\to\infty$ as $k\to\infty$, then
\[
W_k^\star\longrightarrow V^0,
\qquad
U_k\longrightarrow V^0
\qquad\text{as }k\to\infty.
\]
Thus observable-only ranking is asymptotically socially and privately optimal for every nonnegative sequence $(\ell_k)$, regardless of its growth relative to $(s_k)$.
\end{corollary}

\begin{proof}
For every fixed $d>0$, the definition of $\omega_k$ gives
\begin{equation}\label{eq:location-scale-anticoncentration}
\begin{aligned}
\omega_k\left(d\mid z\right)
&=\sup_{a\in\mathbb R}
  \Pr\left\{a\leq\ell_k+s_k\Theta\leq a+d\mid Z=z\right\}\\
&=\sup_{b\in\mathbb R}
  \Pr\left\{b\leq\Theta\leq b+d/s_k\mid Z=z\right\},
\end{aligned}
\end{equation}
where the second equality substitutes $b\triangleq (a-\ell_k)/s_k$.
By Assumption~\ref{ass:conditional-independence}, $\Theta\mid Z=z$ has a density for almost every $z$, so its distribution function is continuous and hence uniformly continuous.
Since $d/s_k\to0$, it follows that $\omega_k(d\mid z)\to0$ for almost every $z$ as $k\to\infty$.
The bound $0\leq\omega_k\leq1$ and dominated convergence then give
\[
\E[\omega_k\left(d\mid Z\right)]\longrightarrow0
\qquad\text{as }k\to\infty,
\]
which establishes Assumption~\ref{ass:levy-anticoncentration}.
Theorem~\ref{thm:strategic-feature-collapse} then gives the stated limits.
\end{proof}

Equation~\eqref{eq:location-scale-anticoncentration} shows that the common shift $\ell_k$ leaves fixed-width concentration unchanged.
Thus increasing $\ell_k$ alone, with $s_k$ fixed, does not imply the anti-concentration condition. Growing dispersion supplies that condition.

\subsection{The remaining value of screening}

Strategic feature collapse establishes that $W_k^\star-V^0\to0$ as $k\to\infty$, but leaves open how quickly the value of report-based screening vanishes.
The following proposition shows that convergence can be arbitrarily slow across admissible distributions of manipulation ability, even when both the location and scale parameters grow linearly in $k$.

\begin{proposition}[Distribution-dependent rates]
\label{prop:collapse-slow-rate}
Along the path $\ell_k=k-1$ and $s_k=k$, there are primitives satisfying Assumptions~\ref{ass:conditional-independence}, \ref{ass:type-monotonicity}, \ref{ass:manipulation-cost}, and \ref{ass:levy-anticoncentration} for which $W_k^\star-V^0$ converges to zero more slowly than every negative power of $k$ as $k\to\infty$.
Moreover, for every prescribed positive sequence $r_k\to0$ as $k\to\infty$, there exists a choice of primitives satisfying these assumptions for which
\[
\lim_{k\to\infty}\frac{W_k^\star-V^0}{r_k}=\infty.
\]
The reference distribution of manipulation ability may depend on the prescribed sequence $(r_k)$, but remains fixed as $k$ varies.
\end{proposition}

The construction offers an additional screening opportunity to applicants excluded by ob\-serv\-able-only ranking. Those who qualify have high latent quality and reference manipulation ability near its upper endpoint.
Although the relevant interval of reference ability has width proportional to $1/k$, its probability mass can decline arbitrarily slowly.
This explains why the baseline assumptions do not provide a common convergence rate across ability distributions.

The next section imposes additional distributional regularity to characterize the remaining value of screening and the mechanisms that attain it.

\subsubsection{Proof of Proposition~\ref{prop:collapse-slow-rate}}

Let $X\sim\operatorname{Unif}[0,1]$ and $Z\sim N(0,1)$ be independent, and let $\Phi$ denote the standard normal cumulative distribution function. Set $\rho=1/2$, $v(z,x)=\Phi(z)+8x$, and $c(u)=u^2$.
Since $\Phi(Z)$ is uniform on $[0,1]$, we have $m(z)=\Phi(z)+4$, and observable-only ranking admits applicants with $\Phi(z)>1/2$.

Take reference ability $\Theta$ independent of $(X,Z)$ and supported on $[1,2]$, with absolutely continuous upper-tail distribution
\[
\Pr\{2-\Theta\leq u\}=G(u),
\qquad 0\leq u\leq1,
\]
where $G(0)=0$ and $G(1)=1$.
We first derive a welfare bound valid for every such $G$, then choose its behavior near zero.
These primitives satisfy the baseline model assumptions.
Along the path $\Theta_k=(k-1)+k\Theta$, Corollary~\ref{cor:collapse-location-scale} gives conditional anti-concentration and $W_k^\star-V^0\to0$ as $k\to\infty$.

Since reference ability $\Theta$ is supported on $[1,2]$, manipulation ability $\Theta_k$ is supported on $[2k-1,3k-1]$.
For applicants with $\Phi(z)<1/2$, offer certain allocation after report $r\geq3k-1+3/4$.
Since $c(u)=u^2$ and the report requirement exceeds $x\in[0,1]$, the minimum cost of meeting it is
\[
c\left(\frac{3k-1+\frac34-x}{\Theta_k}\right)
=
\left(\frac{3k-1+\frac34-x}{(k-1)+k\Theta}\right)^2.
\]
This cost is at most one if and only if
$
3k-1+\frac34-x\leq(k-1)+k\Theta,
$
or, equivalently,
$
k(2-\Theta)\leq x-\frac34.
$
Applicants choose admission when the cost is at most one, with admission selected at indifference.
Thus admission through this route requires $x\geq3/4$.
Using independence and the uniform distributions of $X$ and $\Phi(Z)$, the added allocation mass is
\begin{equation}\label{eq:M_k_1_8}
M_k=\frac12\int_{3/4}^1
G\left(\frac{x-3/4}{k}\right)\,dx
\leq\frac18,
\end{equation}
where the inequality uses $0\leq G\leq1$.
Every applicant admitted through this route has match value $v(z,x)=\Phi(z)+8x\geq6$, while her manipulation cost is at most one.

Restore exact capacity by admitting automatically only those with $\Phi(z)>1/2+M_k$ and rejecting the intervening observable band.
The direct-admission and screening routes admit masses $1/2-M_k$ and $M_k$, respectively.
In the screening group, reporting at the threshold when its cost is at most one, and truthfully otherwise, gives a measurable best response.
Truthful reporting is optimal in the other groups. 
The resulting mechanism is therefore feasible.

Writing $y\triangleq \Phi(z)$, the expected match value of applicants displaced from automatic admission is $m(z)=y+8\E[X\mid Z=z]=y+4$, since $X$ is independent of $Z$ and uniform on $[0,1]$.
Because $\Phi(Z)$ is uniform, the match value removed is
\[
\int_{1/2}^{1/2+M_k}(y+4)\,dy
=\frac92M_k+\frac12M_k^2.
\]
Consequently,
\begin{equation}\label{eq:slow-rate-lower-bound}
\begin{aligned}
W_k^\star-V^0
&\overset{(a)}{\geq}
5M_k-\left(\frac92M_k+\frac12M_k^2\right)\\
&=\frac12M_k-\frac12M_k^2\\
&\overset{(b)}{\geq}\frac14M_k
\overset{(c)}{\geq}\frac1{64}G\left(\frac1{8k}\right).
\end{aligned}
\end{equation}
Here $(a)$ follows because optimal welfare is at least the welfare of the constructed feasible mechanism: additional admissions contribute at least $5M_k$ after manipulation costs, while displaced automatic admissions remove the displayed match value.
Inequality $(b)$ uses $M_k\leq1/8<1/2$ from \eqref{eq:M_k_1_8}.
For $(c)$, restricting the integral defining $M_k$ to $x\in[7/8,1]$ and using monotonicity of $G$ gives $M_k\geq G(1/(8k))/16$.

To prove the second claim, fix any positive sequence $r_k\to0$ as $k\to\infty$.
The idea is to choose $G(1/(8k))$ equal to the square root of the benchmark.
To ensure that these values are consistent with a nondecreasing distribution function, define
\[
\bar r_k\triangleq \sup_{n\geq k}r_n.
\]
Then $0<r_k\leq\bar r_k$ and $\bar r_k\downarrow0$ as $k\to\infty$.
Choose an integer $K$ such that $\bar r_K<1$, and prescribe
\[
G\left(\frac1{8k}\right)=\sqrt{\bar r_k},
\qquad k\geq K.
\]
Define $G$ by linear interpolation on each interval
$
\left[\frac1{8(k+1)},\frac1{8k}\right],
$ for $k\geq K.$
Set $G(0)=0$, and complete the definition on $[1/(8K),1]$ by linear interpolation to $G(1)=1$.

The prescribed values decrease to zero as $k$ increases, so $G$ is nondecreasing and continuous, including at zero.
Its nonnegative slopes define an integrable density: their integrals sum to
\[
\sum_{n=K}^{\infty}
\left(\sqrt{\bar r_n}-\sqrt{\bar r_{n+1}}\right)
+\left(1-\sqrt{\bar r_K}\right)
=1.
\]
Thus $G$ is an absolutely continuous distribution function.
It depends on the prescribed sequence $(r_k)$, but defines one reference distribution that remains fixed across markets.

For every $k\geq K$, \eqref{eq:slow-rate-lower-bound} gives
\[
\begin{aligned}
\frac{W_k^\star-V^0}{r_k}
\geq\frac{G(1/(8k))}{64r_k}=\frac{\sqrt{\bar r_k}}{64r_k} \geq\frac1{64\sqrt{\bar r_k}},
\end{aligned}
\]
where the last inequality uses $r_k\leq\bar r_k$.
The final expression diverges as $k\to\infty$.
Therefore,
\[
\lim_{k\to\infty}\frac{W_k^\star-V^0}{r_k}=\infty,
\]
which establishes the stated limit claim.

For the first claim, apply this construction with $r_k=1/\log(k+1)$.
This sequence is decreasing, so $\bar r_k=r_k$.
The resulting fixed distribution therefore gives, for every $p>0$ and all sufficiently large $k$,
\[
\frac{W_k^\star-V^0}{k^{-p}}
\geq
\frac{k^p}{64\sqrt{\log(k+1)}}.
\]
The right-hand side diverges as $k\to\infty$.
Together with $W_k^\star-V^0\to0$, this proves convergence slower than every negative power of $k$.

\section{Second-Order Optimality: Dual-Track Screening}
\label{sec:second-order}

Observable-only ranking captures the first-order welfare term, but it completely overlooks applicants whose latent quality compensates for weaker observable quality.
This section characterizes the optimal welfare expansion at scale $1/s_k$ and constructs a dual-track mechanism that attains it.
We call the rule dual-track screening.
Under this rule, the direct-admission track admits applicants above a cutoff for observable quality, and the \emph{screening track} offers the applicants a menu of admission probabilities and corresponding thresholds for a function of reports and observable quality.

We use the location--scale family in \eqref{eq:scaled-ability},
\[
\left(\Theta_k\mid Z=z\right)
\stackrel{d}{=}
\left(\ell_k+s_k\Theta\mid Z=z\right),
\qquad \ell_k\geq0,\quad s_k>0,
\]
with $s_k\to\infty$ as $k\to\infty$.
Here $\Theta$ denotes reference manipulation ability, whose distribution remains fixed as $k$ varies, and $\Theta_k$ denotes manipulation ability in market $k$.

We write $\theta_k$ for an applicant's realized ability in market $k$ and $\theta\triangleq(\theta_k-\ell_k)/s_k$ for the corresponding reference ability.
For observable group $z$, let $\overline\theta(z)$ denote the essential supremum of the conditional reference distribution.
The corresponding upper endpoint in market $k$ is
\[
H_k(z)\triangleq \ell_k+s_k\overline\theta(z).
\]
Thus $\theta$ and $\overline\theta(z)$ are measured in reference-ability coordinates, while $\theta_k$ and $H_k(z)$ are measured in market-$k$ ability coordinates.

Theorem~\ref{thm:strategic-feature-collapse} establishes that $W_k^\star-V^0\to0$ as $k\to\infty$.
Proposition~\ref{prop:collapse-slow-rate} shows, however, that the baseline assumptions do not imply a common convergence rate across reference ability distributions, even along the path $\ell_k=k-1$ and $s_k=k$.
Under the additional regularity conditions below, we characterize this gap at scale $1/s_k$ and construct feasible rules satisfying $W_k^\star-W_k=o(s_k^{-1})$ as $k\to\infty$.
We call such a sequence second-order optimal. The second-order analysis uses the following additional regularity conditions.

\begin{assumption}[Additional second-order regularity]
\label{ass:second-order-regularity}
The model primitives satisfy the following conditions.
\begin{itemize}
\item \textbf{Manipulation ability bounded away from zero.} The reference manipulation ability satisfies
$\Theta\geq\underline\theta$ almost surely for some constant
$\underline\theta>0$.

\item \textbf{Uniformly bounded total variation.}
There exists a finite constant $\overline V_f$ such that
$\operatorname{TV}(f(\cdot\mid z))\leq\overline V_f$
for almost every $z$.

\item \textbf{Zero marginal cost at zero.}
The manipulation cost satisfies $c'(0)=0$, and $c$ is twice continuously differentiable on $[0,\infty)$.
\end{itemize}
\end{assumption}

We use a jointly measurable representative of the conditional density whose sections are of bounded variation for almost every $z$. Total variation is computed after extending the density by zero outside its support, so it includes jumps at support boundaries.
Its uniform bound controls local concentrations of manipulation ability without requiring bounded support.
The lower bound on $\Theta$ and the condition $c'(0)=0$ make the selection gains from contracts with vanishing admission probabilities negligible at scale $1/s_k$.

The remaining screening opportunities depend on the density near the conditional upper endpoint $\overline\theta(z)$.
Define its endpoint density by
\begin{equation}\label{eq:second-edge-density}
\beta(z)\triangleq
\begin{cases}
\displaystyle
\lim_{\theta\uparrow\overline\theta(z)}f(\theta\mid z),
& \overline\theta(z)<\infty,\\
0,
& \overline\theta(z)=\infty.
\end{cases}
\end{equation}
Bounded variation ensures that the limit exists. The endpoint density $\beta(z)$ determines the first-order probability mass in short intervals immediately below $\overline\theta(z)$.

\subsection{Dual-track screening}

Dual-track screening combines direct admission based on observable quality with a report-based screening opportunity for selected applicants who would otherwise be rejected.
The employer chooses the screening menu and adjusts the direct-admission cutoff so that the two tracks jointly respect capacity.
The tracks are alternative routes to admission. The screening track uses applicants' strategic reports, while direct admission requires no report.

The employer chooses a measurable set of eligible observable groups
\[
\mathcal S_k\subseteq
\left\{z:m(z)<\lambda,\ \overline\theta(z)<\infty\right\}.
\]
Thus eligible applicants are excluded by observable-only ranking and belong to groups with a finite upper endpoint of manipulation ability.
Groups with $\overline\theta(z)=\infty$ have unbounded conditional ability support and are excluded from the screening track.
Applicants in these groups are admitted through the direct-admission track if their observable quality exceeds the adjusted cutoff defined below. Otherwise, they are rejected. Recall that $H_k(z)=\ell_k+s_k\overline\theta(z)$ is the corresponding endpoint in market $k$.

The employer chooses a finite menu of admission probabilities $0<q_1<\cdots<q_{N_k}=1$, common across eligible groups.
Using the manipulation-cost function in \eqref{eq:manipulation-cost}, set $u_i\triangleq c^{-1}(q_i)$.
Thus $u_i$ is the normalized effort whose cost equals $q_i$. 
Let $u_0\triangleq 0$ and $u_*\triangleq c^{-1}(1)$.
The latter is the largest normalized effort consistent with nonnegative utility, since admission yields a payoff of at most $1$.
Then $0=u_0<u_1<\cdots<u_{N_k}=u_*$.
The menu and its number of contracts may vary with $k$, and we suppress this dependence in $q_i$ and $u_i$.

For each eligible group $z$, the employer chooses a measurable
intercept $a_k(z)$ and offers the contracts
\begin{equation}\label{eq:second-optimal-menu}
\left(q_i,r_{i,k}(z)\right)
\triangleq
\left(c(u_i),\,a_k(z)+H_k(z)u_i\right),
\qquad i=1,\ldots,N_k.
\end{equation}
Higher admission probabilities require higher reports.
These requirements give the hypothetical boundary applicant with $x=a_k(z)$ and $\theta_k=H_k(z)$ utility $q_i-c(u_i)=0$ from every contract.
No positive mass of applicants at this boundary is required.

Consider an eligible applicant with latent quality $x$ and manipulation ability $0<\theta_k\leq H_k(z)$.
The applicant weakly prefers contract $i$ to rejection when its admission benefit covers its manipulation cost.
Using \eqref{eq:manipulation-cost}, this condition becomes
\[
\begin{aligned}
q_i-c\left(
\frac{(r_{i,k}(z)-x)_+}{\theta_k}
\right)\geq0
&\Longleftrightarrow
c\left(
\frac{(a_k(z)+H_k(z)u_i-x)_+}{\theta_k}
\right)\leq c(u_i)\\
&\Longleftrightarrow
(a_k(z)+H_k(z)u_i-x)_+\leq\theta_k u_i\\
&\Longleftrightarrow
x-a_k(z)\geq u_i\left(H_k(z)-\theta_k\right).
\end{aligned}
\]
The first equivalence substitutes $q_i=c(u_i)$ and $r_{i,k}(z)=a_k(z)+H_k(z)u_i$ from \eqref{eq:second-optimal-menu}.
The second uses the strict monotonicity of $c$ in Assumption~\ref{ass:manipulation-cost} and $\theta_k>0$ by Assumption~\ref{ass:second-order-regularity}.
For the last equivalence, $\theta_k u_i>0$ implies that the positive-part inequality is equivalent to $a_k(z)+H_k(z)u_i-x\leq\theta_k u_i$.

The quantity $H_k(z)-\theta_k$ measures how far the applicant's manipulation ability lies below the upper endpoint.
The condition therefore says that latent quality above the intercept must compensate for this ability shortfall.
It determines whether contract $i$ is preferable to rejection, and the applicant must also compare it with the other contracts.

To implement the menu, the employer assigns the highest admission probability whose report requirement is met, or zero if none is met. Thus, for $z\in\mathcal S_k$, the screening rule is
\begin{equation}\label{eq:second-common-schedule}
Q_k\left(z,r\right)
\triangleq
\max\left(
\{0\}\cup
\left\{q_i:r\geq r_{i,k}(z),\ i=1,\ldots,N_k\right\}
\right).
\end{equation}
Because the report requirements increase with $i$, a report satisfying $r_{i,k}(z)\leq r<r_{i+1,k}(z)$ receives probability $q_i$.
Reports below the first requirement receive probability zero, while reports meeting the final requirement receive probability one.
Applicants observe this rule and choose their reports to maximize admission probability minus manipulation cost.

Observable-only ranking already fills capacity: $\Pr\{m(Z)>\lambda\}=\rho$.
Admissions through the screening track must therefore displace an equal mass of direct admissions.
For a given screening design and reporting equilibrium, let $M_k\triangleq \E[A_k\mathbf1\{Z\in\mathcal S_k\}]$ denote the mass admitted through screening.
For designs with $M_k<\rho$, raise the direct-admission cutoff to $\lambda+\delta_k$, choosing $\delta_k\geq0$ so that
\begin{equation}\label{eq:second-capacity-calibration}
\Pr\{\lambda<m(Z)\leq\lambda+\delta_k\}=M_k.
\end{equation}
The maintained assumptions make the distribution of $m(Z)$ continuous, so such a cutoff exists.

Policy~\ref{alg:dual-track} summarizes the complete rule.
\begin{algorithm}[H]
\caption{Dual-track screening}\label{alg:dual-track}
\begin{algorithmic}[1]
\Require Eligible set $\mathcal S_k$; screening schedule $Q_k(z,r)$ from \eqref{eq:second-common-schedule}; and a cutoff adjustment $\delta_k\geq0$ satisfying
\eqref{eq:second-capacity-calibration}.
\Statex \textit{For an applicant with observable quality $z$:}
\If{$m(z)>\lambda+\delta_k$}
    \State Admit with probability one, without a report requirement.
\ElsIf{$z\in\mathcal S_k$}
    \State Upon receiving report $r$, admit with probability $Q_k(z,r)$.
\Else
    \State Reject.
\EndIf
\end{algorithmic}
\end{algorithm}

The cutoff adjustment removes exactly $M_k$ admissions from the bottom of the direct-admission track.
Because $\mathcal S_k$ lies below the original cutoff, screening admissions remain unchanged.
Total admission mass is therefore $(\rho-M_k)+M_k=\rho$.

\begin{remark}
This construction rejects applicants with $\lambda<m(z)\leq\lambda+\delta_k$ and excludes groups with unbounded conditional ability support from screening.
Theorem~\ref{thm:second-order} shows that suitable choices within this class attain second-order optimality.
Broader screening may improve welfare in finite markets, but any improvement over a second-order optimal sequence in this class is $o(s_k^{-1})$ as $k\to\infty$.
\end{remark}

\subsection{The second-order welfare gain}
We now choose the common screening menu and the group-specific intercepts in dual-track screening.

Consider first a single report requirement for certain admission in an eligible group $z$, with a fixed intercept $a$.
For an applicant with $x>a$, participation requires
\[
H_k(z)-\frac{x-a}{u_*}\leq\theta_k\leq H_k(z) \,.
\]
Because admission is certain whenever this condition holds, $\bar A_k(z,x)$ equals the probability of qualifying, averaged over manipulation ability.
Under \eqref{eq:scaled-ability}, the corresponding interval for reference ability has width $(x-a)/(s_k u_*)$.
Thus, for fixed $x$ and $a$,
\[
\begin{aligned}
\bar A_k\left(z,x\right)
&=\int_{\overline\theta(z)-(x-a)_+/(s_k u_*)}^{\overline\theta(z)}
f\left(t\mid z\right)\,dt\\
&=
\underbrace{\beta(z)}_{
\substack{\text{reference density}\\
          \text{at the upper endpoint}}}
\;\times\;
\underbrace{\frac{(x-a)_+}{s_k u_*}}_{
\substack{\text{length of the qualifying}\\
          \text{reference-ability interval}}}
+o(s_k^{-1})
\qquad\text{as }k\to\infty \,,
\end{aligned}
\]
where the approximation uses the endpoint-density limit defining $\beta(z)$.
The positive part accounts for applicants with $x\leq a$, whose average admission probability is zero.

For fixed $x>a$ and $a$, the participation condition keeps $H_k(z)-\theta_k$ bounded.
Since $H_k(z)\to\infty$, normalized effort $(H_k(z)u_*+a-x)/\theta_k$ approaches $u_*$, and manipulation cost approaches $c(u_*)=1$ as $k\to\infty$.
Each screening admission also displaces a direct admission whose expected match value approaches $\lambda$.
The leading expected welfare contribution for latent quality $x$ is therefore
\[
\underbrace{v\left(z,x\right)-\lambda-1}_{
\text{net value per admission}}
\;\times\;
\underbrace{\frac{\beta(z)}{s_k u_*}(x-a)_+}_{
\text{leading admission probability}}.
\]
Averaging this product over $X$ gives the leading contribution
$\beta(z)J_z(a)/(s_k u_*)$, where
\[
J_z(a)
\triangleq \E\left[
(v\left(z,X\right)-\lambda-1)(X-a)_+
\mid Z=z
\right].
\]
Raising the intercept concentrates screening on higher latent qualities, while lowering it reaches more applicants.
Since the factor $\beta(z)/(s_k u_*)$ does not depend on $a$, the intercept choice is governed by $J_z(a)$.
Allowing the employer to leave a group unscreened, define \begin{equation}\label{eq:second-marginal-value}
J_z^\star\triangleq
\displaystyle
\max\left\{0,\sup_{a\in\mathbb R}J_z(a)\right\}
\indc{m(z)<\lambda}\,.
\end{equation}
For $m(z)\geq\lambda$, the right-hand side is defined to be zero without evaluating the supremum.

The menu determines how far below the ability endpoint screening reaches and the admission probabilities offered 
there. 
For a finite partition
$P=(0=u_0<u_1<\cdots<u_N=u_*)$ of $[0,u_*]$,
where $u_*=c^{-1}(1)$ and $N\geq1$, define
\begin{equation}\label{eq:second-cost-coefficient}
K(P)
\triangleq
\sum_{i=1}^N
[c(u_i)-c(u_{i-1})]d_i,
\end{equation}
where
\begin{equation}\label{eq:second-frontier-slopes}
d_i\triangleq
\frac{c'(u_i)-c'(u_{i-1})}
{u_i c'(u_i)-u_{i-1}c'(u_{i-1})},
\qquad i=1,\ldots,N.
\end{equation}
We call $K(P)$ the screening yield of the partition~$P$.
The following proposition shows how $K(P)$ determines the leading welfare gain of a fixed dual-track design.
Its proof derives this coefficient from applicants' contract choices.

\begin{proposition}[Fixed-design performance of dual-track screening]
\label{prop:second-dual-track-performance}
Under Assumptions~\ref{ass:conditional-independence}, \ref{ass:type-monotonicity}, \ref{ass:manipulation-cost}, and \ref{ass:second-order-regularity}, suppose manipulation ability follows \eqref{eq:scaled-ability}, with $s_k\to\infty$.
Fix a finite partition $P=(0=u_0<u_1<\cdots<u_N=u_*)$ whose slopes in \eqref{eq:second-frontier-slopes} satisfy $d_1\geq\cdots\geq d_N$, a bounded measurable intercept function $a$, and a measurable eligible set
\[
\mathcal S\subseteq
\left\{
z:m(z)<\lambda,\ \beta(z)J_z^\star>0,\
J_z(a(z))\geq0
\right\}.
\]
Then there exists a feasible sequence of dual-track screening rules as in Policy~\ref{alg:dual-track}, using this fixed design for all sufficiently large $k$, with $M_k=O(s_k^{-1})$ and $\delta_k\to0$.
Let $W_k^{\mathrm{DT}}(P,a,\mathcal S)$ denote the social welfare of the resulting rule and reporting equilibrium in market $k$.
Then
\begin{equation}\label{eq:second-fixed-design-limits}
\lim_{k\to\infty}s_k
\left(W_k^{\mathrm{DT}}\left(P,a,\mathcal S\right)-V^0\right)
=
K(P)\E\!\left[
\beta(Z)J_Z(a(Z))\mathbf1\{Z\in\mathcal S\}
\right].
\end{equation}
\end{proposition}

The proposition separates the common menu choice, through $K(P)$, from the group-specific intercept and eligibility choices. Let
\[
K_c\triangleq\sup_{P\in\mathcal P}K(P),
\]
where $\mathcal P$ denotes the set of all finite partitions of $[0,u_*]$. Define
\begin{equation}\label{eq:second-welfare-coefficient}
D\triangleq K_c\E[\beta(Z)J_Z^\star].
\end{equation}
The following theorem shows that $D$ is the optimal welfare coefficient at scale $1/s_k$ among all feasible screening rules.
Recall that $W_k^\star$ denotes the optimal social welfare over all feasible screening rules and reporting equilibria in market $k$.
\begin{theorem}[Second-order optimality of dual-track screening]
\label{thm:second-order}
Under Assumptions~\ref{ass:conditional-independence}, \ref{ass:type-monotonicity}, \ref{ass:manipulation-cost}, and \ref{ass:second-order-regularity}, suppose manipulation ability follows \eqref{eq:scaled-ability}, with $s_k\to\infty$.
Then $D<\infty$, and optimal welfare satisfies
\begin{equation}\label{eq:second-value}
W_k^\star =V^0+\frac{D}{s_k}+o(s_k^{-1}).
\end{equation}
There exists a feasible sequence of dual-track screening rules as in Policy~\ref{alg:dual-track}, using in each market a finite menu of admission probabilities common across eligible groups.
Let $W_k^{\mathrm{DT}}$ denote the welfare of the resulting rule in market $k$.
This sequence satisfies $M_k=O(s_k^{-1})$, $\delta_k\to0$, and
\[
W_k^\star-W_k^{\mathrm{DT}}=o(s_k^{-1}).
\]
\end{theorem}

The coefficient $D$ is positive exactly when $\Pr\{\beta(Z)J_Z^\star>0\}>0$, meaning that a positive mass of applicants belongs to observable groups with both positive endpoint density and positive screening value. 
If the conditional ability distribution has unbounded upper support in almost every observable group, then $\beta(z)=0$ in those groups by definition. Consequently, $D=0$, and $W_k^\star-V^0=o(s_k^{-1})$, so observable-only ranking is itself second-order optimal at scale $1/s_k$. The theorem does not rule out positive welfare gains of smaller order.

\prettyref{thm:second-order} allows a finite menu of admission probabilities in the screening track. A single report requirement for certain admission has screening yield $1/u_*$. It attains the maximal screening yield exactly when $K_c=1/u_*$. The following corollary characterizes this equality in terms of the manipulation-cost function.
\begin{corollary}[When one hurdle suffices]
\label{cor:second-single-ordeal}
Under the assumptions and asymptotic regime of Theorem~\ref{thm:second-order}, let $e_c(u)\triangleq uc'(u)/c(u)$ denote the elasticity of manipulation cost for $u>0$. Then
\begin{equation}\label{eq:second-single-condition}
K_c=1/u_*
\quad\Longleftrightarrow\quad
e_c(u)\geq e_c(u_*)
\quad\text{for every }0<u<u_*.
\end{equation}
\begin{itemize}
    \item If $K_c=1/u_*$, there exists a second-order optimal dual-track sequence with a single report requirement $r_k(z)=H_k(z)u_*+a_k(z)$ for certain admission in each eligible group, and direct admission above $\lambda+\delta_k$.

    \item If $\E[\beta(Z)J_Z^\star]>0$, the equality $K_c=1/u_*$ is also necessary for a feasible sequence using a single report requirement for certain admission in each screened group to be second-order optimal.
\end{itemize}

\end{corollary}

The condition requires cost elasticity to attain its
minimum on $(0,u_*]$ at $u_*$.
It holds for power costs satisfying the maintained smoothness assumptions, including quadratic costs, and for some nonhomogeneous costs such as $c(u)=2u^2/(1+u)$.
Strict convexity alone does not imply the condition.
If $e_c(u)<e_c(u_*)$ for some $u\in(0,u_*)$ and $\E[\beta(Z)J_Z^\star]>0$, allowing intermediate admission probabilities strictly increases the best attainable second-order welfare gain relative to a single report requirement for certain admission.

\subsubsection{Proof of Theorem~\ref{thm:second-order}}

Proposition~\ref{prop:second-dual-track-performance} calculates the welfare coefficient of each fixed design.
We now optimize over these designs and bound the performance of arbitrary screening rules. 
Throughout, we work under the assumptions and asymptotic regime of Theorem~\ref{thm:second-order}.

The following proposition supplies the menu properties needed to optimize the fixed-design coefficient.
\begin{proposition}[Menu optimization]
\label{prop:second-cost-frontier}
The maximal screening yield satisfies
\begin{equation}\label{eq:second-cost-bounds}
\frac{1}{u_*}\leq K_c\leq c'(u_*).
\end{equation}
The supremum defining $K_c$ can be approached using partitions whose slopes in \eqref{eq:second-frontier-slopes} are nonincreasing.
\end{proposition}
We next bound the welfare gain within each observable group.
Write
\[
\Sigma(z)\triangleq \operatorname{Cov}\left(v\left(z,X\right),X\mid Z=z\right).
\]
Monotonicity of $v(z,\cdot)$ implies $\Sigma(z)\geq0$, and the second-moment condition in Assumption~\ref{ass:conditional-independence} gives $\E[\Sigma(Z)]<\infty$.
Define the capacity-adjusted payoff of group $z$ by
\[
\mathcal L_{k,z}(Q_k)
\triangleq \E\left[(v\left(z,X\right)-\lambda)A_k-C_k\mid Z=z\right].
\]
Observable-only ranking gives this group payoff
$(m(z)-\lambda)_+$.

\begin{proposition}[Conditional welfare bound]
\label{prop:second-conditional-bound}
Suppose Assumptions~\ref{ass:conditional-independence}, \ref{ass:type-monotonicity}, \ref{ass:manipulation-cost}, and \ref{ass:second-order-regularity} hold.
Under \eqref{eq:scaled-ability}, with $s_k\to\infty$ as $k\to\infty$, every sequence of admissible screening rules and reporting equilibria satisfies, for almost every $z$,
\begin{equation}\label{eq:second-global-dominator}
s_k\left[
\mathcal L_{k,z}(Q_k)-(m(z)-\lambda)_+
\right]
\leq K_0\Sigma(z)
\qquad\text{for every }k,
\end{equation}
and
\begin{equation}\label{eq:second-conditional-bound}
\limsup_{k\to\infty}
s_k\left[
\mathcal L_{k,z}(Q_k)-(m(z)-\lambda)_+
\right]
\leq\beta(z)K_cJ_z^\star,
\end{equation}
where $K_0$ is a finite constant depending only on $c$, $\underline\theta$, and $\overline V_f$, uniformly over markets, observable groups, and screening rules.
Moreover, $0\leq J_z^\star\leq\Sigma(z)$ almost surely and $D<\infty$.
\end{proposition}

\begin{proof}[Proof of Theorem~\ref{thm:second-order}]
For any feasible sequence with welfare $W_k$, exact capacity gives
\[
W_k-V^0
=\E\left[\mathcal L_{k,Z}(Q_k)-(m(Z)-\lambda)_+\right].
\]
Proposition~\ref{prop:second-conditional-bound} bounds the scaled integrand above by the integrable function $K_0\Sigma(Z)$. Reverse Fatou's inequality and the proposition's conditional bound therefore give
\[
\limsup_{k\to\infty}s_k(W_k-V^0)\leq D.
\]
Choosing feasible rules with $W_k\geq W_k^\star-s_k^{-2}$ yields
\[
\limsup_{k\to\infty}s_k(W_k^\star-V^0)\leq D.
\]

For the lower bound, we construct dual-track rules whose scaled welfare gain converges to $D$.
Because the design will change across markets, we first bound screening admissions by a multiple of $1/s_k$ that does not depend on the chosen design.
Fix a design $(P,a,\mathcal S)$ satisfying the conditions of Proposition~\ref{prop:second-dual-track-performance}, and set $L_c\triangleq\max_{0\leq u\leq u_*}c''(u)<\infty$.
For $z\in\mathcal S$, suppose an applicant selects a positive contract $j$. The participation condition gives $x-a(z)\geq u_j[H_k(z)-\theta_k]$.
For each $i=1,\ldots,j$, we have $u_i\leq u_j$, so
\[
0\leq H_k(z)-\theta_k
\leq\frac{(x-a(z))_+}{u_i}.
\]
Her selected admission probability satisfies
\[
c(u_j)=\sum_{i=1}^{j}[c(u_i)-c(u_{i-1})],
\]
since $c(u_0)=0$. Thus the increment $c(u_i)-c(u_{i-1})$ contributes exactly when a contract $j\geq i$ is selected.
Rejected applicants contribute zero.

Conditional independence and the ability-density bound $\overline V_f/s_k$ therefore give
\[
\begin{aligned}
s_k\bar A_k\left(z,x\right)
&\leq\overline V_f(x-a(z))_+
\sum_{i=1}^N
\frac{c(u_i)-c(u_{i-1})}{u_i}\\
&\leq\overline V_f L_cu_*(x-a(z))_+ \,,
\end{aligned}
\]
where the first inequality bounds the probability of selecting a contract $j\geq i$ by the density bound times the width $(x-a(z))_+/u_i$, and the second inequality uses $c'(0)=0$ and $c'(t)\leq L_ct$, which imply $c(u_i)-c(u_{i-1})\leq L_cu_i(u_i-u_{i-1})$.

Also, $J_z(a(z))\geq0$ on $\mathcal S$ implies
\[
\begin{aligned}
[\lambda-m(z)+1]\E\left[(X-a(z))_+\mid Z=z\right]
&\leq\operatorname{Cov}\left(
v\left(z,X\right),(X-a(z))_+\mid Z=z
\right)\\
&\leq\Sigma(z) \,,
\end{aligned}
\]
where the last inequality follows from the covariance identity, since the positive-part function is nondecreasing and $1$-Lipschitz.
Because $\lambda-m(z)+1>1$ on $\mathcal S$, averaging the allocation bound over latent quality and eligible observable groups gives
\[
s_kM_k
\leq\overline V_f L_cu_*
\E[\Sigma(Z)\mathbf1\{Z\in\mathcal S\}]
\leq\overline V_f L_cu_*\E[\Sigma(Z)]
\triangleq C_M<\infty.
\]
The constant $C_M$ does not depend on $P$, $a$, $\mathcal S$, or $k$.

Now let $n=1,2,\ldots$ index a sequence of fixed designs. By Proposition~\ref{prop:second-cost-frontier},
choose partitions $P_n$ with nonincreasing slopes and $K(P_n)\to K_c$.
For intercepts, use nested finite rational grids whose union is dense in $\mathbb R$. On the $n$th grid, maximize $J_z(a)$ and let $\mathcal S_n$ consist of groups with $m(z)<\lambda$, $\beta(z)>0$, and a positive maximum.
Fixed tie-breaking gives a measurable maximizing intercept $a_n(z)$ on $\mathcal S_n$. Set $a_n(z)=0$ elsewhere.
Then $\mathcal S_n$ is measurable and $a_n$ is bounded.
By continuity of $J_z(a)$ in $a$, the weighted selected payoffs increase pointwise to $\beta(z)J_z^\star$. 
Monotone convergence and $K(P_n)\to K_c$ give
\[
D_n\triangleq K(P_n)\E\left[
\beta(Z)J_Z(a_n(Z))\mathbf1\{Z\in\mathcal S_n\}
\right]\longrightarrow D.
\]

For each fixed design,
Proposition~\ref{prop:second-dual-track-performance} gives feasible rules with welfare
$W_k^{\mathrm{DT}}(P_n,a_n,\mathcal S_n)$ satisfying
\[
s_k\left(W_k^{\mathrm{DT}}\left(P_n,a_n,\mathcal S_n\right)-V^0\right)
\longrightarrow D_n.
\]
Choose strictly increasing market indices $k_1<k_2<\cdots$ such that design $n$ is feasible and, for every $k\geq k_n$,
\[
\left|s_k\left(W_k^{\mathrm{DT}}\left(P_n,a_n,\mathcal S_n\right)-V^0\right)-D_n\right|
\leq\frac1n.
\]
Use design $n$ in markets $k_n\leq k<k_{n+1}$ and observable-only ranking for $k<k_1$.
Let $W_k^{\mathrm{DT}}$ and $M_k$ denote the resulting welfare and screening admission mass. Then
\[
s_k(W_k^{\mathrm{DT}}-V^0)\longrightarrow D,
\qquad s_kM_k\leq C_M.
\]
In particular, $M_k=O(s_k^{-1})$. 
The capacity equation and the continuity and strict increase of the distribution function of $m(Z)$ near $\lambda$ give $\delta_k\to0$. Each rule has the required menu and eligibility properties by construction.
Since $W_k^{\mathrm{DT}}\leq W_k^\star$, the upper bound and the constructed sequence prove \eqref{eq:second-value} and $W_k^\star-W_k^{\mathrm{DT}}=o(s_k^{-1})$.
\end{proof}

\subsection{An illustration of dual-track screening}

This example illustrates two implications of the preceding results. Screening can have positive value at scale $1/s_k$ even when the population ability distribution is unbounded, because the relevant endpoints are conditional on observable quality. 
Holding the type distribution fixed, the manipulation cost function determines whether a single report requirement suffices or intermediate admission probabilities improve welfare.

Let $X\sim\operatorname{Unif}[0,1]$ and $Z\sim N(0,1)$ be independent, and let $\Phi$ denote the standard normal cumulative distribution function.
Set $v(z,x)=\Phi(z)+4x$ and $\rho=1/2$.
Conditional on $Z=z<0$, let reference ability be uniform on $[1,2]$. 
Conditional on $Z=z\geq0$, let $\Theta=1+E$, where $E$ is exponentially distributed with mean one. In both groups, ability is independent of $X$ conditional on $Z$.
The population ability distribution is unbounded, while the conditional densities have a common finite total-variation bound.

Since $\Phi(Z)$ is uniform on $[0,1]$, we have $m(z)=\Phi(z)+2$, $\lambda=5/2$, and $V^0=11/8$.
For excluded groups $z<0$, $\beta(z)=1$ and maximization in \eqref{eq:second-marginal-value} gives
\[
a_z=\frac34-\frac{\Phi(z)}2,
\qquad
J_z^\star=\frac13\left(\frac14+\frac{\Phi(z)}2\right)^3,
\qquad
\E[\beta(Z)J_Z^\star]=\frac5{512}.
\]

For $c(u)=u^2$, we have $u_*=1$ and $K_c=1$.
The screening track therefore offers the single hurdle
\[
r_k(z)=\ell_k+2s_k+\frac34-\frac{\Phi(z)}2
\qquad(z<0),
\]
while the direct-admission track uses the slightly raised cutoff for $m(z)$. 
Optimal welfare satisfies
\[
W_k^\star=\frac{11}{8}+\frac5{512s_k}+o(s_k^{-1})
\qquad\text{as }k\to\infty.
\]

For $c(u)=(u^2+u^4)/2$, a common menu still suffices, but a single hurdle does not.
The partition $P=\{0,1/2,1\}$ offers admission probabilities $5/32$ and $1$ and has
\[
K(P)
=\frac{5/32}{1/2}
+\frac{27}{32}\frac{3-3/4}{3-3/8}
=\frac{29}{28}>1.
\]
Hence $K_c\geq29/28>1$, confirming the gain from intermediate admission probabilities.
The two report requirements are $(\ell_k+2s_k)/2+a_z$ and $\ell_k+2s_k+a_z$.
The first contract attracts additional applicants who do not choose the full-admission contract.
Both requirements use the same intercept $a_z$.

\subsection{A numerical illustration of the mechanisms}\label{subsec:numerical-preview}

A simple example illustrates the value and limitations of using manipulable reports in hiring.
Suppose an applicant's observable quality $Z$ and latent quality $X$ are independent and uniformly distributed on $[0,1]$.
The employer values an applicant at $Z+4X$ and hires half the applicants.
Applicants can report a higher latent quality at a cost that depends on privately known manipulation ability, which is independent of both productive qualities.
Specifically, inflating a report from $x$ to $r>x$ costs $(r-x)^2/(2\theta^2)$, where $\theta$ is uniform on $[0.5,1]$; reporting $r\leq x$ is costless.
Figures~\ref{fig:numerical-preview}(a)--(d) compare four allocations.\footnote{Although the employer does not observe $X$, we use it as the horizontal axis to show which underlying applicant types are selected.
Each point shows hiring probability at a given pair $(x,z)$, averaged over manipulation ability.}

\begin{figure}[!htbp]
\centering
\refstepcounter{figure}
\label{fig:numerical-preview}

\begin{minipage}[t]{0.48\textwidth}
\centering
\includegraphics[width=\linewidth]{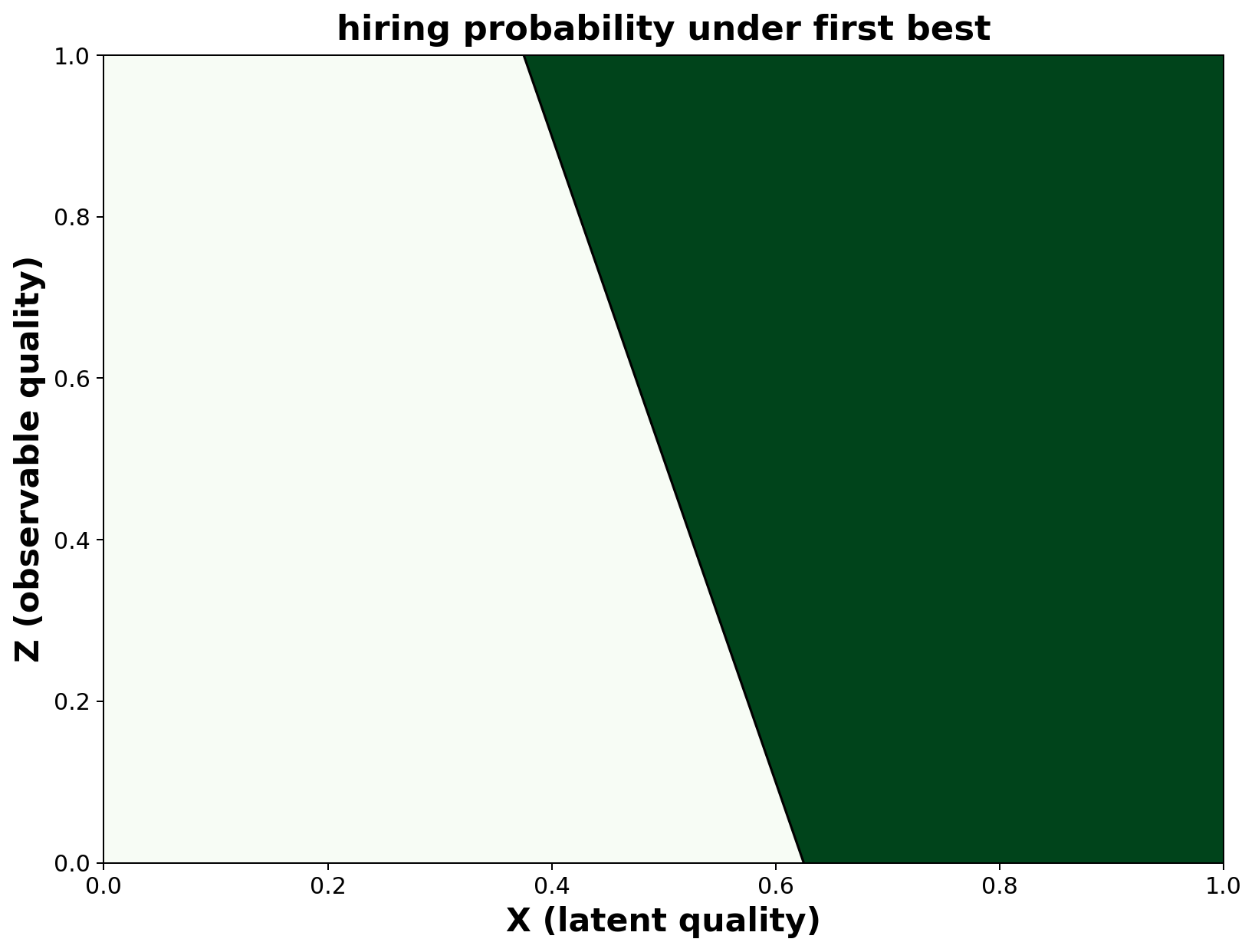}
\par\smallskip
{\footnotesize \textbf{Figure~\thefigure(a).} First best: the ideal allocation if both qualities were observable, which is unattainable because latent quality is private information.\par}
\end{minipage}\hfill
\begin{minipage}[t]{0.48\textwidth}
\centering
\includegraphics[width=\linewidth]{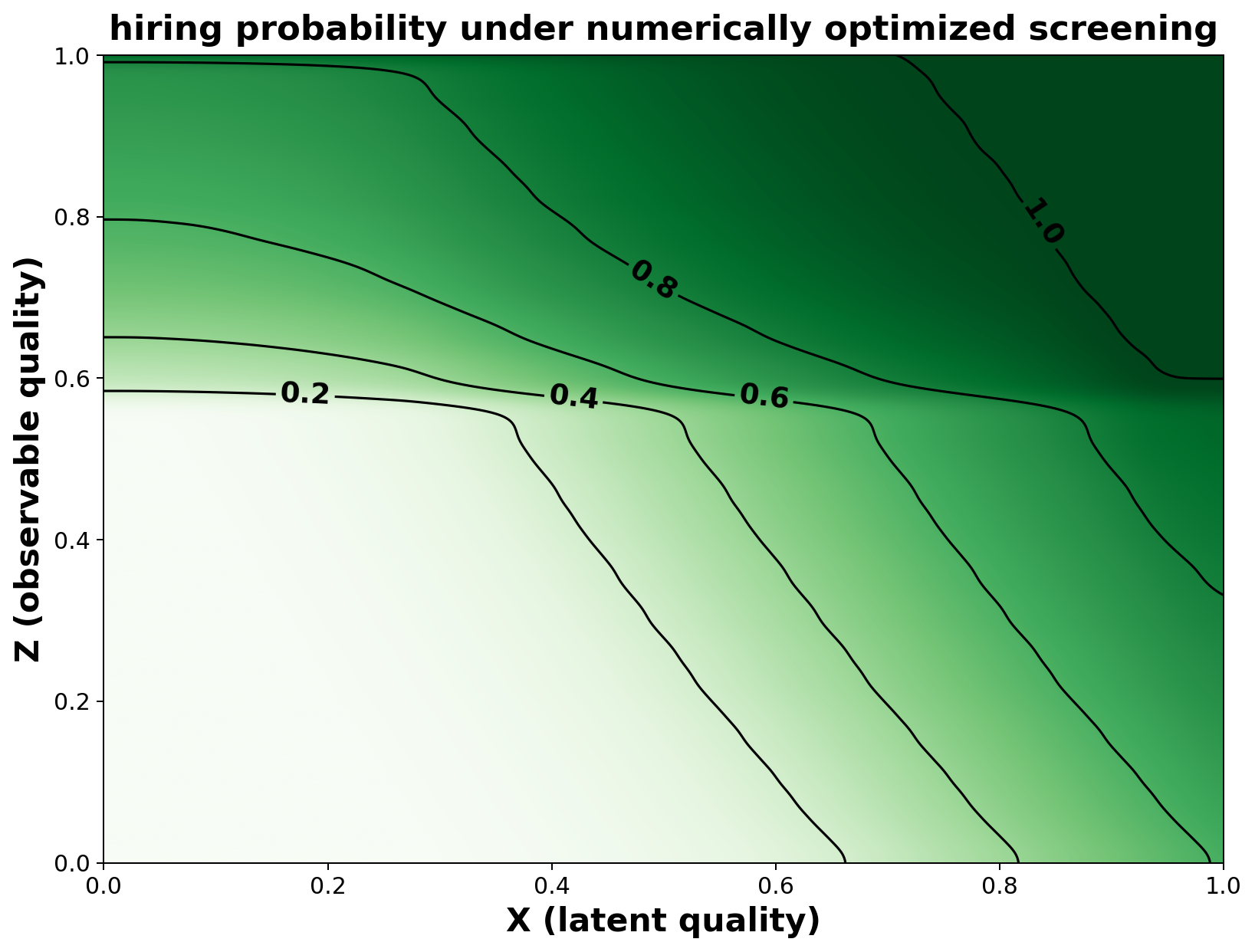}
\par\smallskip
{\footnotesize \textbf{Figure~\thefigure(b).} Numerically optimized screening (approximate second best): relative to first best, some applicants with high latent quality but low observable quality are rejected, while others with low latent quality but high observable quality are admitted.\par}
\end{minipage}

\medskip

\begin{minipage}[t]{0.48\textwidth}
\centering
\includegraphics[width=\linewidth]{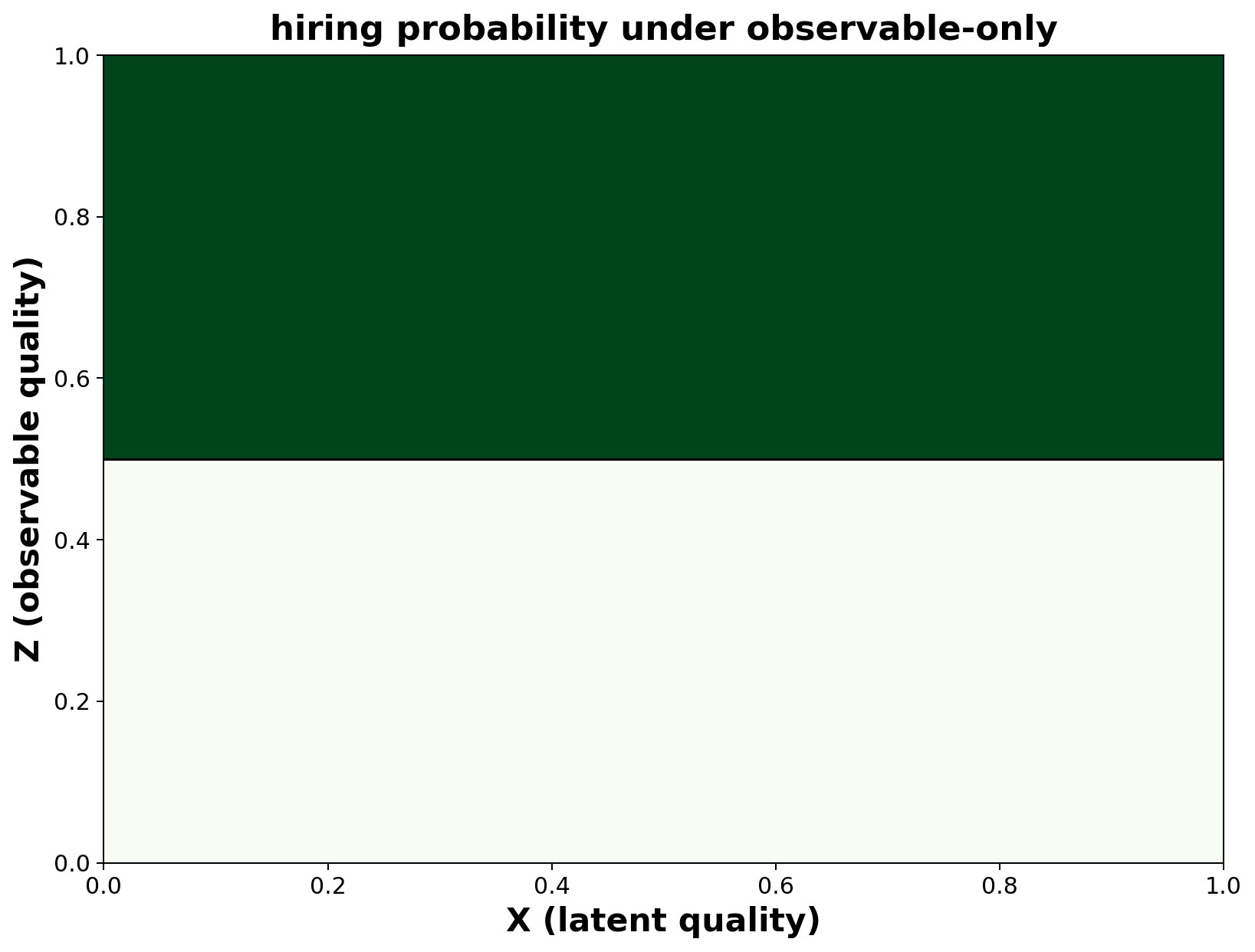}
\par\smallskip
{\footnotesize \textbf{Figure~\thefigure(c).} Observable-only: hiring depends entirely on observable quality, regardless of latent quality.\par}
\end{minipage}\hfill
\begin{minipage}[t]{0.48\textwidth}
\centering
\includegraphics[width=\linewidth]{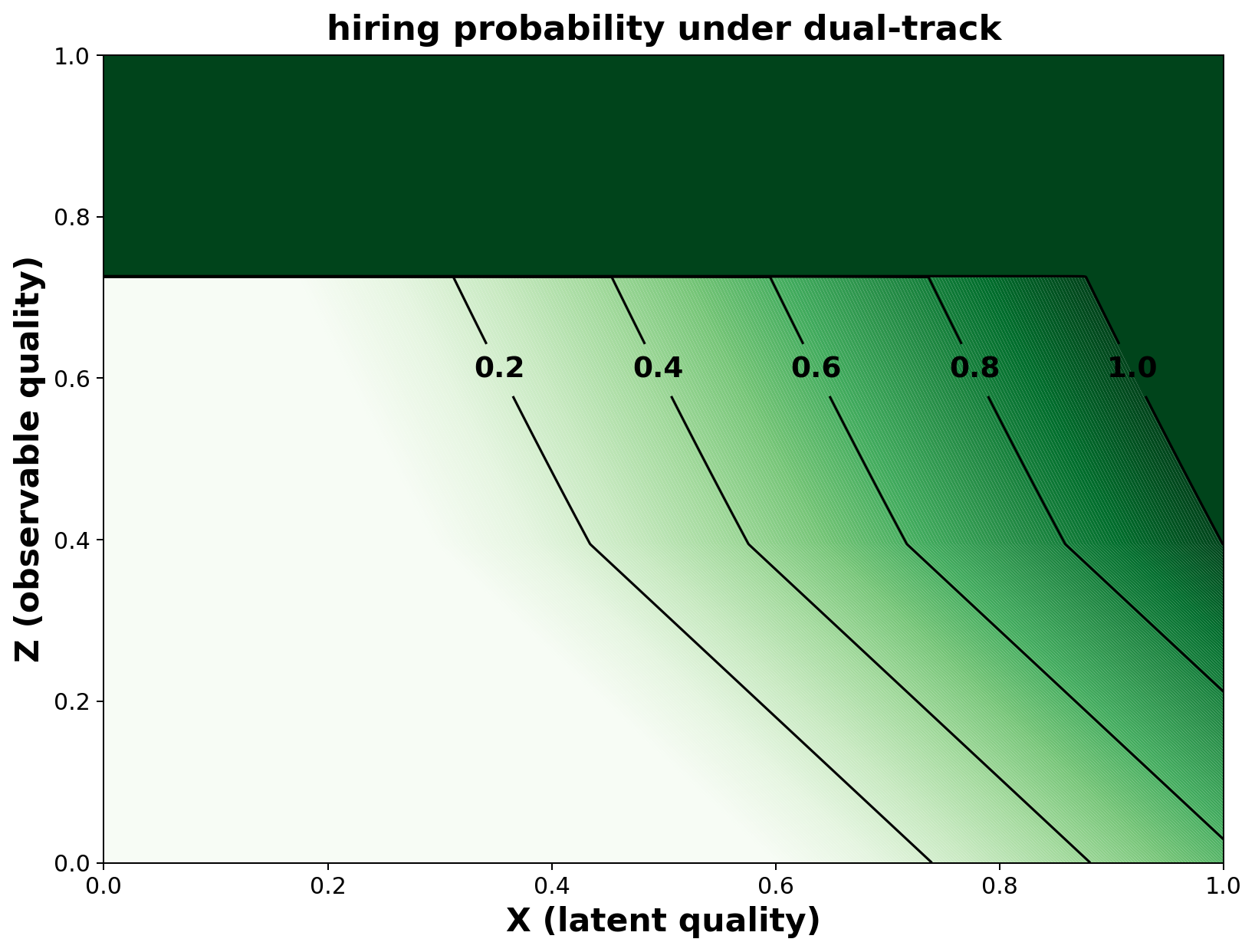}
\par\smallskip
{\footnotesize \textbf{Figure~\thefigure(d).} Dual-track: applicants with high observable quality are admitted directly, while others can qualify through a report threshold that identifies some applicants with high latent but low observable quality. The resulting hiring pattern resembles that under numerically optimized screening.
\par}
\end{minipage}

\end{figure}

The numerically optimized screening rule (upper right) partially restores selection on latent quality.\footnote{We approximate unrestricted screening by optimizing over finite menus of report requirements and possibly randomized admission probabilities, using nonlinear optimization.}
The dual-track rule (lower right) offers a much simpler alternative, directly admitting applicants with high observable quality and offering others a report-based route to admission.\footnote{For this finite-market illustration, we optimize over a broader class of dual-track mechanisms than Policy~1. The report-based track extends up to the raised direct-admission cutoff. This additional flexibility can improve welfare in finite markets, although the more restrictive mechanism in Policy~1 already achieves second-order optimality.}
These patterns illustrate how reports can improve targeting compared with observable-only hiring.

\section{The External Cost of Private Screening}\label{sec:externality}

The preceding sections evaluate screening by match value net of
applicants' manipulation costs. A private employer maximizes match
value alone. Theorem~\ref{thm:strategic-feature-collapse} gives $W_k^\star\to V^0$ and $U_k\to V^0$, but applicants' costs under privately optimal rules need not vanish. We characterize the resulting welfare loss and construct an economy with a sequence of exactly privately optimal rules whose costs approach applicants' full aggregate admission benefit.

Recall that $\pi_k$ is the employer's expected match value and
$W_k=\pi_k-\E[C_k]$ is normalized social welfare under a specified
feasible rule and reporting equilibrium.
\begin{proposition}[Private screening and its external cost]
\label{prop:private-externality}
Under the assumptions of Theorem~\ref{thm:strategic-feature-collapse},
every feasible rule and reporting equilibrium satisfies
$\E[C_k]\leq\rho$ and
\begin{equation}\label{eq:externality-upper}
0\leq W_k^\star-W_k\leq\rho+\left(U_k-\pi_k\right).
\end{equation}
For every feasible sequence with $U_k-\pi_k\to0$,
\begin{equation}\label{eq:private-collapse-necessity}
W_k^\star-W_k=\E[C_k]+o(1)\qquad\text{as }k\to\infty.
\end{equation}
In particular, $\limsup_{k\to\infty}(W_k^\star-W_k)\leq\rho$.
\end{proposition}

\begin{proof}
Reporting $r=x$ is costless and gives nonnegative utility.
Optimal reporting therefore implies $C_k\leq A_k$ almost surely,
so exact capacity gives $\E[C_k]\leq\E[A_k]=\rho$.
The identity
$W_k^\star-W_k=W_k^\star-\pi_k+\E[C_k]$ and
$W_k^\star\leq U_k$ from Proposition~\ref{prop:collapse-rate}
give \eqref{eq:externality-upper}.
Theorem~\ref{thm:strategic-feature-collapse} gives
$W_k^\star\to V^0$ and $U_k\to V^0$.
If $U_k-\pi_k\to0$, then $\pi_k\to V^0$ as well. 
Substituting these limits into the same identity proves
\eqref{eq:private-collapse-necessity}.
\end{proof}

Each admission is worth one to its recipient, so $\rho$ is the applicants' aggregate benefit from admission before manipulation costs. 
The next theorem attains this bound in the limit with privately optimal rules that select the same applicants as costless observable-only ranking.
\begin{theorem}[The private externality can be maximal]
\label{thm:maximal-private-externality}
For every $\rho\in(0,1)$, there exist primitives satisfying Assumptions~\ref{ass:conditional-independence}, \ref{ass:type-monotonicity}, \ref{ass:manipulation-cost}, and \ref{ass:second-order-regularity}, with unbounded conditional ability support and $v(z,\cdot)$ strictly increasing for almost every $z$, such that $U_k=W_k^\star=V^0$ in every market under \eqref{eq:scaled-ability}.
Along any such path with $s_k\to\infty$ and $\ell_k/s_k\to\infty$, there exists a feasible sequence of screening rules and reporting equilibria that is exactly privately optimal, implements $A_k=a^0(Z)$ almost surely in every market, and satisfies
\begin{equation}\label{eq:maximal-private-externality}
\E[C_k]\longrightarrow\rho,
\qquad
W_k^\star-W_k\longrightarrow\rho
\qquad\text{as }k\to\infty.
\end{equation}
Costless observable-only ranking is also exactly privately optimal in every market.
\end{theorem}

\begin{proof}
Fix $\rho\in(0,1)$ and take mutually independent $Z\sim N(0,1)$, $X\sim\operatorname{Unif}[0,1]$, and $E\sim\operatorname{Exp}(1)$. 
Set $\Theta\triangleq1+E$.
Let $\Phi$ denote the standard normal cumulative distribution function and set
\[
v\left(z,x\right)
\triangleq \Phi(z)+\left|\Phi(z)-\left(1-\rho\right)\right|
                 \left(x-\tfrac12\right).
\]
Then $m(z)=\Phi(z)$, and uniformity of $\Phi(Z)$ on $[0,1]$ gives $\lambda=1-\rho$. 
The density of $Z$ is continuous and positive on $\mathbb R$, and $v(z,\cdot)$ is strictly increasing except at the single cutoff group. 
The reference ability density has lower endpoint one, unbounded upper support, and total variation two after extension by zero. 
Any cost function satisfying Assumption~\ref{ass:manipulation-cost} and $c'(0)=0$  can be used if it is twice continuously differentiable on $[0,\infty)$.
Boundedness of $X$ and $v(Z,X)$ gives the required second moments.
These primitives satisfy all the stated assumptions.

For $x\in[0,1]$, the latent-quality term has absolute value at most $|m(z)-\lambda|/2$. Thus $v(z,x)>\lambda$ when $m(z)>\lambda$, and $v(z,x)<\lambda$ when $m(z)<\lambda$.
For any allocation $A\in[0,1]$ with $\E[A]=\rho$, exact capacity gives
\[
\E\left[v\left(Z,X\right)A\right]-V^0
=
\E\left[
\left(v\left(Z,X\right)-\lambda\right)
\left(A-a^0(Z)\right)
\right]\leq0.
\]
Observable-only ranking attains this match value at zero cost, so $U_k=W_k^\star=V^0$ in every market.

Recall that $u_*=c^{-1}(1)$. Impose the report requirement
$r\geq u_*\ell_k$ on the applicants selected by observable-only ranking:
\[
Q_k\left(z,r\right)
\triangleq
\mathbf1\left\{m(z)>\lambda,\ r\geq u_*\ell_k\right\}.
\]
For every type, a measurable best response is obtained by comparing the utilities from reports $x$ and $\max\{x,u_*\ell_k\}$, with a fixed tie-breaking order.
For applicants in the population, $X\geq0$ and $\Theta_k=\ell_k+s_k\Theta>\ell_k$ almost surely.
Meeting the requirement therefore takes normalized effort less than $u_*$ and costs less than $c(u_*)=1$ in every market.
Thus $A_k=a^0(Z)$ almost surely, and the rule is feasible and exactly privately optimal. Its manipulation cost is
\[
C_k=a^0(Z)c\left(
\frac{\left(u_*\ell_k-X\right)_+}{\ell_k+s_k\Theta}
\right).
\]

Along the stated path, the common improvement grows faster than the dispersion scale: $\ell_k\to\infty$ and $s_k/\ell_k\to0$.
For all sufficiently large $k$ and every fixed $x$ and $\theta>0$,
\[
\frac{\left(u_*\ell_k-x\right)_+}{\ell_k+s_k\theta}
=
\frac{\left(u_*-x/\ell_k\right)_+}
     {1+\left(s_k/\ell_k\right)\theta}
\longrightarrow u_*.
\]
Continuity of $c$ and $0\leq C_k\leq a^0(Z)$ therefore give $\E[C_k]\to\E[a^0(Z)]=\rho$ by dominated convergence.
Since $\pi_k=V^0=W_k^\star$ in every market, $W_k^\star-W_k=\E[C_k]$, proving \eqref{eq:maximal-private-externality}.
\end{proof}

The costly rule and costless observable-only ranking select the same applicants and generate the same match value. The employer is therefore indifferent between them.
Choosing the costless implementation eliminates the welfare loss in this example without changing who is admitted.

\bibliographystyle{chicago}
\bibliography{references}

\clearpage
\appendix

\section{Notation table}

\begingroup
\small
\renewcommand{\arraystretch}{1.15}
\begin{longtable}{@{}>{\raggedright\arraybackslash}p{0.23\textwidth}
                  >{\raggedright\arraybackslash}p{0.72\textwidth}@{}}
\caption{Main notation}\label{tab:notation}\\
\toprule
Symbol & Meaning \\
\midrule
$Z$, $X$
& Observable quality and latent quality; $z$ and $x$ denote their realized values.\\
$\Theta$, $\Theta_k$
& Manipulation ability and ability in market $k$.
  Under the location--scale specification, $\Theta$ is reference ability.\\
$\ell_k$, $s_k$
& Location and scale parameters:
  $(\Theta_k\mid Z=z)\stackrel{d}{=}
  (\ell_k+s_k\Theta\mid Z=z)$.\\
$v(z,x)$
& Match value of an applicant with qualities $(z,x)$.\\
$m(z)$
& Expected match value given observable quality:
  $m(z)=\mathbb{E}[v(z,X)\mid Z=z]$.\\
$c(u)$; $u_*$
& Cost of normalized manipulation effort $u$;
  $u_*=c^{-1}(1)$ is the effort whose cost equals one.\\
$\rho$
& Admission capacity, with $\mathbb{E}[A_k]=\rho$.\\
$Q_k(z,r)$
& Admission probability assigned in market $k$ to an applicant with observable quality $z$ and report $r$.\\
$A_k$, $C_k$
& Admission probability and manipulation cost induced by  a specified rule and reporting equilibrium in market $k$.\\
$\bar A_k(z,x)$
& Admission probability averaged over manipulation ability,
  holding observable and latent quality fixed at $(z,x)$.\\
$\lambda$
& Expected-match-value cutoff for observable-only ranking:
  $\Pr\{m(Z)>\lambda\}=\rho$.\\
$a^0(z)$
& Observable-only allocation:
  $a^0(z)=\mathbf{1}\{m(z)>\lambda\}$.\\
$V^0$
& Welfare of observable-only ranking:
  $V^0=\mathbb{E}[m(Z)a^0(Z)]$.\\
$\pi_k$
& Employer's expected match value under a specified rule
  and equilibrium: $\pi_k=\mathbb{E}[v(Z,X)A_k]$.\\
$U_k$
& Supremum of the employer's expected match value over
  feasible rules and reporting equilibria in market $k$.\\
$W_k$
& Normalized social welfare under a specified rule and
  equilibrium: $W_k=\pi_k-\mathbb{E}[C_k]$.\\
$W_k^\star$
& Supremum of normalized social welfare over feasible
  rules and reporting equilibria in market $k$.\\
$W_k^{\mathrm{DT}}$
& Normalized social welfare of the specified dual-track
  rule and reporting equilibrium in market $k$.\\
$\overline\theta(z)$; $H_k(z)$
& Conditional essential upper endpoint of reference ability,
  and the market endpoint
  $H_k(z)=\ell_k+s_k\overline\theta(z)$.\\
$\beta(z)$
& Reference density at a finite conditional upper endpoint;
  zero for an infinite endpoint.\\
$J_z(a)$; $J_z^\star$
& $J_z(a)=\mathbb E[(v(z,X)-\lambda-1)(X-a)_+\mid Z=z]$;
  $J_z^\star$ is its nonnegative optimal value in excluded
  groups and zero elsewhere.\\
$K(P)$; $K_c$
& Screening yield of an effort partition $P$;
  $K_c=\sup_{P\in\mathcal P}K(P)$.\\
$D$
& Optimal welfare coefficient at scale $1/s_k$:
  $D=K_c\mathbb E[\beta(Z)J_Z^\star]$.\\
\bottomrule
\end{longtable}
\endgroup

\addtocontents{toc}{\protect\setcounter{tocdepth}{1}}

\section{Omitted proofs in \prettyref{sec:second-order}}

\subsection{Proof of Proposition~\ref{prop:second-dual-track-performance}}

We first calculate the allocation generated by a fixed menu.

Fix a group $z$ with a finite ability endpoint, an intercept $a$, and a finite partition
\[
P=\left(0=u_0<u_1<\cdots<u_N=u_*\right),
\]
all independently of $k$. 
Contract $i$ offers admission probability $c(u_i)$ at report requirement $a+H_k(z)u_i$.

For fixed latent quality $x$ and endpoint distance $t\geq0$, set $\theta_k=H_k(z)-t$.
For all sufficiently large $k$, this ability is positive and each positive contract's report requirement exceeds $x$.
Its normalized manipulation is then
\[
\frac{H_k(z)u_i+a-x}{\theta_k}
=
u_i+\frac{u_i t-(x-a)}{\theta_k}.
\]
Taylor expansion and $H_k(z)/\theta_k\to1$ give
\begin{equation}\label{eq:second-local-contract}
H_k(z)\left[
c(u_i)-
c\left(\frac{H_k(z)u_i+a-x}{\theta_k}\right)
\right]
\longrightarrow
c'(u_i)(x-a-u_i t)
\qquad\text{as }k\to\infty.
\end{equation}
Multiplying all utilities by $H_k(z)>0$ preserves their ranking.
For this fixed finite menu, the limiting payoffs and the zero outside option therefore identify the selected contract for all sufficiently large $k$ whenever the limiting maximizer is unique.

The next lemma uses these limiting choices to calculate the screening yield defined in
\eqref{eq:second-cost-coefficient}.

\begin{lemma}[Allocation under a fixed menu]
\label{lem:second-menu-allocation}
Fix an intercept $a$ and a finite partition $P$ whose slopes in \eqref{eq:second-frontier-slopes} satisfy $d_1\geq\cdots\geq d_N$.
Let $q(x,t)$ denote the admission probability selected from the limiting utilities in \eqref{eq:second-local-contract}, with rejection available.
Then
\begin{equation}\label{eq:second-integrated-tangent}
\int_0^\infty q\left(x,t\right)\, \mathrm{d}t
=
K(P)(x-a)_+.
\end{equation}
\end{lemma}

Fix a design $(P,a,\mathcal S)$ satisfying the proposition, and use the screening menu in \eqref{eq:second-optimal-menu} on $\mathcal S$.
Each applicant need only compare the report thresholds above $x$ with the costless report $x$.
A fixed tie-breaking order therefore gives a measurable best response, as required by Definition~\ref{def:reporting-equilibrium}.

For $z\in\mathcal S$, participation in contract $i$ requires $H_k(z)-\Theta_k\leq(x-a(z))/u_i$.
Thus positive allocation is confined, up to a null boundary event, to $x>a(z)$ and the ability band $0\leq H_k(z)-\Theta_k\leq(x-a(z))/u_1$.
The density bound implies
\[
0\leq s_k\bar C_k\left(z,x\right)\leq s_k\bar A_k\left(z,x\right)
\leq\frac{\overline V_f}{u_1}(x-a(z))_+.
\]
For fixed $x,z$, the density of the endpoint distance $t\triangleq H_k(z)-\Theta_k$ is $f(\overline\theta(z)-t/s_k\mid z)/s_k$.
On the preceding fixed band, its scaled value converges to $\beta(z)$ for $t>0$.
By \eqref{eq:second-local-contract}, at each fixed $t$ outside finitely many switching boundaries, allocation and cost both converge to the limiting selected admission probability.
Dominated convergence and
\eqref{eq:second-integrated-tangent} therefore give
\begin{equation}\label{eq:second-recovery-allocation}
\begin{aligned}
\lim_{k\to\infty}s_k\bar A_k\left(z,x\right)
&=\lim_{k\to\infty}s_k\bar C_k\left(z,x\right)\\
&=\beta(z)K(P)(x-a(z))_+,
\qquad z\in\mathcal S.
\end{aligned}
\end{equation}
The same limits can be integrated over $(X,Z)$: bounded intercepts and the second moments make
$(|v(Z,X)-\lambda|+1)(|X|+\|a\|_\infty)$ integrable.
In particular,
\[
\lim_{k\to\infty}s_kM_k
=K(P)\E\left[
\beta(Z)(X-a(Z))_+\mathbf1\{Z\in\mathcal S\}
\right]<\infty,
\]
so $M_k=O(s_k^{-1})$.

For all sufficiently large $k$, choose $\delta_k\geq0$ to satisfy \eqref{eq:second-capacity-calibration}, and admit directly when $m(z)>\lambda+\delta_k$.
The distribution function of $m(Z)$ is continuous and strictly increasing near $\lambda$, so such a choice exists
and $\delta_k\to0$ as $k\to\infty$.
Because $\mathcal S$ lies below $\lambda$, screening admissions are unchanged and total admission is exactly
$\rho$.
The match value displaced beyond $\lambda M_k$ satisfies
\[
0\leq
\E\left[(m(Z)-\lambda)
\mathbf1\{\lambda<m(Z)\leq\lambda+\delta_k\}\right]
\leq M_k\delta_k.
\]
Since $M_k=O(s_k^{-1})$ and $\delta_k\to0$, this loss
is $o(s_k^{-1})$ as $k\to\infty$.
Write $W_k^{\mathrm{DT}}(P,a,\mathcal S)$ for the welfare of this calibrated rule. Exact capacity gives
\[
\begin{aligned}
W_k^{\mathrm{DT}}\left(P,a,\mathcal S\right)-V^0
&=\E\left[
\left((v\left(Z,X\right)-\lambda)A_k-C_k\right)
\mathbf1\{Z\in\mathcal S\}\right]\\
&\quad-\E\left[(m(Z)-\lambda)
\mathbf1\{\lambda<m(Z)\leq\lambda+\delta_k\}\right].
\end{aligned}
\]
Using \eqref{eq:second-recovery-allocation} in the first term and the preceding displacement bound in the second proves \eqref{eq:second-fixed-design-limits}, by the definition of $J_z(a(z))$.
Use observable-only ranking in any remaining initial markets.
This proves the proposition.
\subsubsection{Proof of Lemma~\ref{lem:second-menu-allocation}}

Write $\Delta c_i\triangleq c(u_i)-c(u_{i-1})$ and $\Delta h_i\triangleq u_i c'(u_i)-u_{i-1}c'(u_{i-1})>0$.
By \eqref{eq:second-local-contract}, adjacent limiting payoffs differ by
\[
\Delta h_i[d_i(x-a)-t],
\]
where contract $0$ represents rejection, using $c'(0)=0$.

For $x>a$, the ordering $d_1\geq\cdots\geq d_N$ makes these differences change sign at most once, from positive to negative, as $i$ increases.
Away from the finitely many switching distances, if contract $j$ is selected, precisely the first $j$ differences are positive.
For $x\leq a$, rejection is optimal for every $t>0$. 
Consequently, for almost every $t\geq0$,
\[
q\left(x,t\right)=\sum_{i=1}^N
\Delta c_i\,\mathbf1\{t<d_i(x-a)_+\}.
\]
Integrating over $t$ gives
\[
\int_0^\infty q\left(x,t\right)\, \mathrm{d}t
=(x-a)_+\sum_{i=1}^N\Delta c_i d_i
=K(P)(x-a)_+,
\]
which proves \eqref{eq:second-integrated-tangent}.

\subsection{Proof of Proposition~\ref{prop:second-cost-frontier}}

Write $q(u)\triangleq c(u)$, $p(u)\triangleq c'(u)$, and $h(u)\triangleq uc'(u)$,
all zero at $u=0$, and let $\Delta$ denote increments over partition intervals.
Since $(h-q)'(u)=uc''(u)\geq0$,
\[
K(P)
=\sum_i\frac{\Delta q_i\Delta p_i}{\Delta h_i}
\leq\sum_i\Delta p_i
=p(u_*).
\]
The partition $\{0,u_*\}$ gives $1/u_*$, proving
\eqref{eq:second-cost-bounds}.

Suppose two adjacent intervals meeting at $t$ have slopes $d_1<d_2$, and let $\widetilde{P}$ remove $t$.
Then
\[
\begin{aligned}
\frac{\Delta q_1}{\Delta h_1}
&\geq1-td_1>1-td_2
\geq\frac{\Delta q_2}{\Delta h_2},\\
K(\widetilde{P})-K(P)
&=\frac{
(\Delta q_1\Delta h_2-\Delta q_2\Delta h_1)(d_2-d_1)
}{\Delta h_1+\Delta h_2}>0.
\end{aligned}
\]
The first line uses $\Delta h-\Delta q=\int uc''(u)\,du$, with $u\leq t$ on the first interval and $u\geq t$ on the second.
The second line follows by expanding the partition values. Repeating this deletion gives nonincreasing slopes without decreasing $K(P)$.

\subsection{Proof of Proposition~\ref{prop:second-conditional-bound}}

Write $B_k\triangleq A_k-C_k$ for applicant utility. Costless truthful reporting implies $B_k\geq0$, and hence $0\leq C_k\leq A_k\leq1$.
Recall that $\bar A_k(z,x)$ averages allocation over
$\Theta_k\mid Z=z$.
Similarly, let $\bar C_k(z,x)$ denote average manipulation cost and let
\[
\bar B_k(z,x)\triangleq \bar A_k(z,x)-\bar C_k(z,x)
\]
denote average applicant utility.
Write $b_z(x)\triangleq v(z,x)-\lambda-1$.
We first use a bound that is uniform over screening rules.

\begin{lemma}[Uniform allocation sensitivity]
\label{lem:second-sensitivity}
Set
\[
M_c\triangleq \max_{0\leq u\leq u_*}u c''(u),
\qquad
K_0\triangleq 2\overline V_fc'(u_*)+\frac{M_c}{\underline\theta}.
\]
For every admissible screening rule, almost every $z$, and $x<x'$,
\begin{equation}\label{eq:second-lipschitz}
0\leq\bar A_k(z,x')-\bar A_k(z,x)
\leq\frac{K_0}{s_k}(x'-x).
\end{equation}
The function $\bar B_k(z,\cdot)$ is nondecreasing and Lipschitz with constant
$c'(u_*)/(s_k\underline\theta)$.
Consequently,
\begin{equation}\label{eq:second-covariance-bound}
0\leq I_k(z)
\triangleq\E[(v(z,X)-m(z))A_k\mid Z=z]
\leq\frac{K_0\Sigma(z)}{s_k}.
\end{equation}
\end{lemma}
The constant in this lemma supplies $K_0$ in the proposition.
Since
\[
\mathcal L_{k,z}(Q_k)
=I_k(z)+(m(z)-\lambda)\E[A_k\mid z]-\E[C_k\mid z],
\]
$0\leq A_k\leq1$ and $C_k\geq0$ imply
\eqref{eq:second-global-dominator}.

The following refinement gives the sharp coefficient for sequences whose conditional payoff is within order $1/s_k$ of the observable-only benchmark from below.

\begin{lemma}[Local allocation limits]
\label{lem:second-localization}
For almost every $z$, consider a sequence of admissible
screening rules $(Q_k)$.
\begin{itemize}
    \item If $m(z)<\lambda$ and $s_k\mathcal L_{k,z}(Q_k)$ is bounded below, then $s_k\bar A_k(z,x)$ is bounded uniformly on each compact $x$-interval.
    Under these conditions, for every fixed $x<x'$,
    \begin{equation}\label{eq:second-sharp-comparison}
    \begin{aligned}
    \limsup_{k\to\infty}
    s_k[\bar A_k(z,x')-\bar A_k(z,x)]
    &\leq\beta(z)K_c(x'-x),\\
    s_k[\bar B_k(z,x')-\bar B_k(z,x)]
    &\longrightarrow0
    \qquad\text{as }k\to\infty.
    \end{aligned}
    \end{equation}
    \item If $m(z)>\lambda$ and $s_k[\mathcal L_{k,z}(Q_k)-(m(z)-\lambda)]$ is bounded below, then, for every fixed $x<x'$,
    \[
    s_k[\bar A_k(z,x')-\bar A_k(z,x)]
    \longrightarrow0
    \qquad\text{as }k\to\infty.
    \]
\end{itemize}

\end{lemma}

Fix a sequence of admissible rules and reporting equilibria, and an observable group $z$ outside the conditional null sets.
Suppose first that $m(z)<\lambda$, and take a subsequence along which $s_k\mathcal L_{k,z}(Q_k)$ approaches its upper limit. If this limit is $-\infty$, the desired bound is immediate. Otherwise, the scaled payoffs are bounded below along this subsequence, so Lemma~\ref{lem:second-localization} applies. Relabel the subsequence by $k$ and define the scaled average cost
\[
\widehat C_k(x)\triangleq s_k\bar C_k(z,x)
      =s_k[\bar A_k(z,x)-\bar B_k(z,x)].
\]
Since $0\leq\widehat C_k(x)\leq s_k\bar A_k(z,x)$, these functions are uniformly bounded on compact intervals.
The Lipschitz bounds for $\bar A_k$ and $\bar B_k$ give a common Lipschitz constant for $\widehat C_k$, so a further subsequence converges uniformly on compact intervals to a nonnegative function $\widehat C$.
By \eqref{eq:second-sharp-comparison}, utility differences vanish at this scale, so $\widehat C$ is nondecreasing and Lipschitz with constant $\beta(z)K_c$.

Since $A_k=C_k+B_k$ and $b_z(x)=v(z,x)-\lambda-1$,
\[
\begin{aligned}
s_k\mathcal L_{k,z}(Q_k)
&=\E[b_z(X)\widehat C_k(X)\mid z]
  +[m(z)-\lambda]s_k\bar B_k(z,0)\\
&\quad+\E\left[(v(z,X)-\lambda)s_k
      [\bar B_k(z,X)-\bar B_k(z,0)]\mid z\right]\\
&\leq\E[b_z(X)\widehat C_k(X)\mid z]+o(1)
\qquad\text{as }k\to\infty.
\end{aligned}
\]
The second term is nonpositive. The last term tends to zero by the utility comparison and dominated convergence: $s_k|\bar B_k(z,x)-\bar B_k(z,0)| \leq c'(u_*)|x|/\underline\theta$.
The common linear-growth bound for $\widehat C_k$ likewise gives $\E[b_z(X)\widehat C_k(X)\mid z]\to\E[b_z(X)\widehat C(X)\mid z]$ as $k\to\infty$.
Both uses of dominated convergence follow from the conditional second moments.

To bound this limiting payoff, write $L\triangleq \beta(z)K_c$.
If $L=0$, then $\widehat C$ is a nonnegative constant and
$\E[b_z(X)\widehat C(X)\mid z]\leq0$ because $\E[b_z(X)\mid z]<0$.
The same bound holds if $b_z(X)\leq0$ almost surely.
Otherwise, monotonicity and the negative mean of $b_z$ give a finite threshold $\xi\triangleq\inf\{x:b_z(x)>0\}$.
Set $a\triangleq\xi-\widehat C(\xi)/L$.
Nonnegativity and the Lipschitz bound imply $\widehat C(x)\geq L(x-a)_+$ for $x\leq\xi$ and $\widehat C(x)\leq L(x-a)_+$ for $x\geq\xi$.
The two functions agree at $\xi$, so multiplying by $b_z(x)$ gives
\[
\begin{aligned}
\E[b_z(X)\widehat C(X)\mid z]
&\leq L\E[b_z(X)(X-a)_+\mid z]\\
&\leq LJ_z^\star=\beta(z)K_cJ_z^\star.
\end{aligned}
\]
This proves the conditional upper bound for excluded groups.

For $m(z)>\lambda$, again take a subsequence realizing the upper limit of the scaled excess payoff. Unless this limit is $-\infty$, Lemma~\ref{lem:second-localization} applies. 
For an independent copy $X'$ conditional on $Z=z$, the conditional covariance identity gives
\[
s_k I_k(z)
=\frac12\E\left[
(v(z,X)-v(z,X'))s_k
[\bar A_k(z,X)-\bar A_k(z,X')]\mid z\right]
\longrightarrow0
\qquad\text{as }k\to\infty.
\]
Here the allocation difference tends to zero after scaling, and \eqref{eq:second-lipschitz} bounds the integrand in absolute value by $K_0|v(z,X)-v(z,X')|\,|X-X'|$, which is integrable.
Since $\mathcal L_{k,z}(Q_k)-(m(z)-\lambda)\leq I_k(z)$, the conditional upper bound is zero.
The event $m(Z)=\lambda$ has probability zero because $Z$ is atomless and $m$ is strictly increasing.

The coefficient is finite. For an excluded group, set $Y_a\triangleq (X-a)_+$. Its monotonicity and unit Lipschitz constant
give
\[
\begin{aligned}
\E[b_z(X)Y_a\mid z]
&=\operatorname{Cov}(v(z,X),Y_a\mid z)
  +[m(z)-\lambda-1]\E[Y_a\mid z]\\
&\leq\operatorname{Cov}(v(z,X),Y_a\mid z)
\leq\Sigma(z),
\end{aligned}
\]
where the last inequality follows from the pairwise covariance identity.
Taking the supremum over $a$ gives $0\leq J_z^\star\leq\Sigma(z)$ in excluded groups.
For $m(z)\geq\lambda$, the same bound holds because $J_z^\star=0$ by definition.
Continuity of $J_z(a)$ in $a$ allows its supremum to be taken over rational intercepts, so $J_Z^\star$ is measurable.
Finally, the density variation bound gives $\beta(z)\leq\overline V_f$, and
Proposition~\ref{prop:second-cost-frontier} gives $K_c\leq c'(u_*)$. Hence
\[
0\leq D
=K_c\E[\beta(Z)J_Z^\star]
\leq c'(u_*)\overline V_f\E[\Sigma(Z)]<\infty.
\]
This completes the proof of the proposition.

\subsubsection{Proof of Lemma~\ref{lem:second-sensitivity}}

Fix $k$ and an observable group $z$ for which the assumptions hold. 
Write $b$ for market-$k$ ability, with density $g_k(b\mid z)\triangleq s_k^{-1}f((b-\ell_k)/s_k\mid z)$.
Let $B(x,b)$ denote optimal applicant utility.
Since reporting $r=x$ is costless, $B\geq0$ and every selected contract has effort $u\leq u_*$.
Replacing contract utilities by their positive parts leaves $B$ unchanged. These positive parts have $x$-slopes bounded by $c'(u_*)/b$ and ability-slopes bounded by $u_*c'(u_*)/b$.
Their supremum is therefore locally Lipschitz in ability, and nondecreasing and Lipschitz in $x$. Averaging and using $b\geq s_k\underline\theta$ proves the stated bound for $\bar B_k$.

First consider a finite menu including rejection, and remove dominated contracts so that reports and allocations are strictly ordered. Lemma~\ref{lem:collapse-comparison} makes the selected allocation nondecreasing in ability.
At a switch from $(q_0,r_0)$ to $(q_1,r_1)$, put $u_i\triangleq (r_i-x)_+/b$, $p_i\triangleq c'(u_i)$, and $h_i\triangleq u_i c'(u_i)$, all zero for rejection. Writing $\Delta$ for the change between contracts, indifference and differentiation give
\begin{equation}\label{eq:second-switch-derivative}
\Delta q=c(u_1)-c(u_0),
\qquad
\frac{db}{dx}=-\frac{\Delta p}{\Delta h}<0.
\end{equation}
Here $\Delta h>0$, and $c'(0)=0$ handles the positive part
at $r_i=x$. These pairwise thresholds are locally Lipschitz, extending by zero when the higher report becomes costless. Each allocation-level cutoff is a finite combination of their minima and maxima. Since $g_k$ is bounded, $\bar A_k(z,\cdot)$ is locally absolutely continuous, with
\begin{equation}\label{eq:second-allocation-derivative}
\bar A_k'(z,x)
=\sum_j g_k(b_j\mid z)\Delta q_j
\frac{\Delta p_j}{\Delta h_j}
\end{equation}
for almost every $x$. Density discontinuities exclude only a null set of $x$: there are countably many, and each strictly decreasing pairwise threshold meets each positive ability level at most once.

Since $[u c'(u)-c(u)]'=u c''(u)\geq0$, we have $\Delta h_j\geq\Delta q_j$. Fix $x$ outside the preceding null set and write $p(b)\triangleq c'(u(b))$ for the selected effort.
Between switches, $p'(b)=-u(b)c''(u(b))/b$. At switches, $p$ jumps upward by $\Delta p_j$.
Integration by parts therefore gives
\begin{equation}\label{eq:second-weighted-variation}
\begin{aligned}
\sum_j g_k(b_j\mid z)\Delta p_j
&=-\int_0^\infty p(b)\,d g_k(b\mid z)
  +\int_0^\infty
  \frac{g_k(b\mid z)u(b)c''(u(b))}{b}\,db\\
&\leq c'(u_*)\operatorname{TV}(g_k(\cdot\mid z))
  +M_c\int_0^\infty\frac{g_k(b\mid z)}b\,db.
\end{aligned}
\end{equation}
Zero extension includes support-boundary jumps, while $g_k=0$ near zero and $g_k\to0$ at infinity eliminate endpoint terms. The inequality uses $p\leq c'(u_*)$ and $u c''(u)\leq M_c$.
Since $\operatorname{TV}(g_k)\leq\overline V_f/s_k$ and $\int_0^\infty g_k(b\mid z)b^{-1}\,db \leq(s_k\underline\theta)^{-1}$,
\eqref{eq:second-allocation-derivative} gives $0\leq\bar A_k'(z,x)\leq K_0/s_k$.
Integrating proves \eqref{eq:second-lipschitz} for finite menus.

For an arbitrary admissible rule, take increasing finite submenus dense in its contract image, including rejection.
Fix $x$ and write $B_n(b)\uparrow B(b)$ for their utilities. This convergence follows from continuity of the cost.
The local Lipschitz property makes $B$ differentiable at almost every $b$. At such a $b$, take any convergent subsequence of selected efforts, $u_n\to u$ in $[0,u_*]$.
For every $t>0$, comparing the same contracts at ability $t$ and taking the limit along this subsequence gives
\[
B(t)\geq B(b)+c(u)-c\left(\frac bt u\right).
\]
Equality at $t=b$ implies $bB'(b)=u c'(u)$. Strict monotonicity of $u\mapsto u c'(u)$ determines a unique limit, and the same comparison identifies it with the effort under every optimal report in the full rule.
Thus $A_n(b)=B_n(b)+c(u_n)$ converges to the original allocation for almost every ability.
The ability distribution has a density and $0\leq A_n\leq1$, so dominated convergence passes the averaged allocations, and their common bound, to the original rule for every fixed $x$.

Finally, let $X'$ be an independent copy conditional on $Z=z$. The conditional symmetrization used in \eqref{eq:collapse-information}, monotonicity, and \eqref{eq:second-lipschitz} give
\[
0\leq I_k(z)
\leq\frac{K_0}{2s_k}\E\left[
(v(z,X)-v(z,X'))(X-X')\mid Z=z\right]
=\frac{K_0\Sigma(z)}{s_k},
\]
where the equality is the conditional covariance identity.

\subsubsection{Proof of Lemma~\ref{lem:second-localization}}

Fix an observable group $z$ outside the conditional null sets with $m(z)\ne\lambda$, and first consider finite menus.
Write $B_k(x,b)$ for applicant utility at market-$k$ ability $b$.
Under the corresponding payoff hypothesis, the conditional payoff identity and Lemma~\ref{lem:second-sensitivity} give
\[
\begin{aligned}
&|m(z)-\lambda|\E[|A_k-a^0(z)|\mid z]+\E[C_k\mid z]\\
&\qquad=I_k(z)-\left[\mathcal L_{k,z}(Q_k)
                  -(m(z)-\lambda)_+\right]
=O(s_k^{-1}).
\end{aligned}
\]

Suppose first that $m(z)<\lambda$. Mean allocation is $O(s_k^{-1})$, and \eqref{eq:second-lipschitz} gives
\[
s_k\bar A_k(z,x)
\leq s_k\E[A_k\mid z]+K_0\E[|X-x|\mid z],
\]
which is bounded uniformly on compact $x$-intervals.
For $0<\eta<1$, put $u_\eta\triangleq c^{-1}(\eta)$.
Writing $u_k\triangleq u_k(x,\Theta_k)$ for selected effort, the utility
envelope formula, $b\geq s_k\underline\theta$, and $c(u_k)\leq A_k$ imply
\[
0\leq s_k\bar B_k'(z,x)
\leq\frac{\E[c'(u_k)\mid z]}{\underline\theta}
\leq\frac{1}{\underline\theta}
\left[c'(u_\eta)
      +\frac{c'(u_*)}{\eta}\bar A_k(z,x)\right]
\]
for almost every $x$, where the last inequality splits according to whether $c(u_k)\leq\eta$. 
Integrating over $[x,x']$ and taking $k\to\infty$ followed by $\eta\downarrow0$ proves the utility limit in \eqref{eq:second-sharp-comparison}.

To bound allocation sensitivity, retain only contracts with allocation at most $\eta$, including rejection. Every original switch whose upper allocation is at most $\eta$ remains a switch in this restricted menu, whose selected efforts are at most $u_\eta$.
Equations~\eqref{eq:second-allocation-derivative} and \eqref{eq:second-weighted-variation} therefore bound these switches' contribution to $s_k\bar A_k'(z,x)$ by
\[
2\overline V_f c'(u_\eta)
+\frac{\max_{0\leq u\leq u_\eta}u c''(u)}
       {\underline\theta},
\]
which tends to zero as $\eta\downarrow0$.

At every remaining switch $b_j$, monotonicity in ability gives $\Pr\{\Theta_k\geq b_j\mid z\}\leq\bar A_k(z,x)/\eta
=O(s_k^{-1})$.
Thus, among switches with nonzero density weight, reference abilities approach the upper endpoint uniformly on compact $x$-intervals.
If that endpoint is infinite, apply \eqref{eq:second-weighted-variation} with the density restricted to reference abilities above $T$.
Both $\operatorname{TV}(f\mathbf1_{[T,\infty)})$ and $\int_T^\infty f(\theta\mid z)/\theta\,d\theta$ tend to zero as $T\to\infty$: finite variation and integrability imply that the density and its tail variation vanish.
Taking $k\to\infty$ and then $T\to\infty$ shows that the remaining contribution vanishes.

For a finite endpoint, write $H_k\triangleq \ell_k+s_k\overline\theta(z)$.
For sufficiently small fixed $e>0$, monotonicity of allocation
and the ability-slope bound for utility give
\[
0\leq B_k(x,H_k)
\leq\frac{\bar A_k(z,x)}
              {\Pr\{\Theta>\overline\theta(z)-e\mid z\}}
   +\frac{u_*c'(u_*)e}{\underline\theta}.
\]
Here the first term bounds utility at $H_k-s_ke$, and the second bounds its increase over the remaining ability interval.
The denominator is positive by the definition of the upper endpoint. Taking $k\to\infty$ and then $e\downarrow0$ shows that $B_k(x,H_k)\to0$ uniformly on compact $x$-intervals.

At each remaining contributing switch, its common utility $B_j$ therefore tends to zero, its efforts are $c^{-1}(q_j-B_j)$ and $c^{-1}(q_{j-1}-B_j)$, and $s_kg_k(b_j\mid z)\to\beta(z)$.
For upper allocation exceeding $\eta$, the ratio $\Delta c'/\Delta[uc']$ consequently converges uniformly to its zero-utility counterpart. This follows from its continuous extension at equal positive efforts, with value $c''(u)/(c'(u)+u c''(u))$. The upper effort stays bounded away from zero.
For fixed $\eta>0$, this continuous extension is uniformly continuous on the relevant compact set. Each ratio error is therefore at most $\epsilon_k\to0$, and the total error is bounded by $\epsilon_k\sum_j\Delta q_j\leq\epsilon_k$, independently of menu size.
Complete these allocation intervals to a partition of $[0,1]$ and apply $c^{-1}$ to its endpoints. This gives an effort partition in $\mathcal P$, so their contribution is at most $\beta(z)K_c+o(1)$.
Combining both parts, integrating over $[x,x']$, and taking $k\to\infty$ followed by $\eta\downarrow0$ proves the allocation bound in \eqref{eq:second-sharp-comparison}.

Now suppose that $m(z)>\lambda$. The opening identity gives $\E[B_k\mid z]\to1$, and the utility Lipschitz bound implies $\bar B_k(z,x)\to1$ uniformly on compact $x$-intervals.
Choose a finite $T\geq\underline\theta$ with $\Pr\{\Theta\leq T\mid z\}>0$. Monotonicity in ability implies $B_k(0,\ell_k+s_kT)\to1$. A contract $(q_k,r_k)$ selected there has $q_k\to1$ and $u_k=(r_k)_+/(\ell_k+s_kT)\to0$.
For every fixed $R>0$, this same report satisfies
\[
\frac{(r_k-x)_+}{b}
\leq\frac{T}{\underline\theta}u_k
    +\frac{R}{s_k\underline\theta}
\]
for $|x|\leq R$ and $b\geq\ell_k+s_k\underline\theta$. The right-hand side tends to zero as $k\to\infty$, so this contract provides utility tending uniformly to one throughout the set. Every selected cost is at most $1-B_k(x,b)$, so selected efforts tend uniformly to zero.
Applying \eqref{eq:second-weighted-variation} over this shrinking effort range and integrating proves $s_k[\bar A_k(z,x')-\bar A_k(z,x)]\to0$ as $k\to\infty$.

Finally, in the finite-menu approximation from the proof of Lemma~\ref{lem:second-sensitivity}, effort convergence also gives cost convergence. Fubini's theorem gives almost-everywhere convergence under the conditional joint distribution of $(X,\Theta_k)$. Domination by $|v(z,X)-\lambda|+1$ then gives convergence in conditional payoff. The common Lipschitz bounds also give uniform convergence on compact $x$-intervals in averaged allocation and utility.
In market $k$, choose the payoff error and both uniform errors on $[-k,k]$ to be at most $1/(ks_k)$.
All preceding bounds are independent of menu size, so the conclusions pass to arbitrary admissible rules.

\subsection{Proof of Corollary~\ref{cor:second-single-ordeal}}

We first prove \eqref{eq:second-single-condition}.
By Proposition~\ref{prop:second-cost-frontier}, $K_c\geq1/u_*$, and partitions with nonincreasing slopes suffice for the supremum defining $K_c$.

Write $h(u)\triangleq uc'(u)$. Suppose $e_c(u)\geq e_c(u_*)$ for every $0<u<u_*$.
Since $c(u_*)=1$, this condition is equivalent to $c(u)\leq h(u)/h(u_*)$.
For a partition with nonincreasing slopes, summation by parts gives
\[
\begin{aligned}
K(P)
&=c(u_*)d_N+\sum_{i<N}c(u_i)(d_i-d_{i+1})\\
&\leq
\frac{
h(u_*)d_N+\sum_{i<N}h(u_i)(d_i-d_{i+1})
}{h(u_*)}\\
&=\frac{c'(u_*)}{h(u_*)}
=\frac1{u_*}.
\end{aligned}
\]
The inequality uses $d_i-d_{i+1}\geq0$. The next equality follows from $[h(u_i)-h(u_{i-1})]d_i=c'(u_i)-c'(u_{i-1})$
and $c'(0)=0$.
Taking the supremum and using the lower bound gives $K_c=1/u_*$.

Conversely, if the condition fails at some $u\in(0,u_*)$, the partition $\{0,u,u_*\}$ satisfies
\begin{equation}\label{eq:second-two-contract-gain}
K(\left\{0,u,u_*\right\})-\frac1{u_*}
=
\frac{
(u_*-u)c(u)[e_c(u_*)-e_c(u)]
}{
u_*u[u_*c'(u_*)-uc'(u)]
}>0.
\end{equation}
Thus $K_c>1/u_*$, proving
\eqref{eq:second-single-condition}.

We next establish the welfare implications. Keep the single-hurdle partition $\{0,u_*\}$ fixed in the construction of
Proposition~\ref{prop:second-dual-track-performance}. Choosing intercepts and eligible groups, and then a sequence of designs, as in the proposition and the proof of Theorem~\ref{thm:second-order}, gives feasible rules satisfying
\[
s_k(W_k^{\mathrm{DT}}-V^0)
\longrightarrow
\frac{\E[\beta(Z)J_Z^\star]}{u_*}
\qquad\text{as }k\to\infty.
\]
If the elasticity condition holds, this limit equals $D$, so Theorem~\ref{thm:second-order} implies $W_k^\star-W_k^{\mathrm{DT}}=o(s_k^{-1})$ as $k\to\infty$.

For necessity, consider a finite menu offering only rejection and certain admission. Only the lowest-report certain-admission contract can be undominated.
At any switch from rejection to admission, utility is zero, so the two efforts are $0$ and $u_*$.
Consequently,
\[
\frac{\Delta c'}{\Delta[uc']}
=
\frac{c'(u_*)}{u_*c'(u_*)}
=
\frac1{u_*}.
\]
The localization argument and the proof of Proposition~\ref{prop:second-conditional-bound} therefore apply with $K_c$ replaced by $1/u_*$.
Finite submenus preserve binary allocations, so the approximation in Lemma~\ref{lem:second-sensitivity} extends the bound to arbitrary admissible binary rules.
Every feasible sequence of such rules consequently satisfies
\[
\limsup_{k\to\infty}s_k(W_k-V^0)
\leq
\frac{\E[\beta(Z)J_Z^\star]}{u_*}.
\]
If $\E[\beta(Z)J_Z^\star]>0$ and the elasticity condition fails,
\eqref{eq:second-two-contract-gain} gives $K_c>1/u_*$.
The binary-rule bound is then strictly below
$D=K_c\E[\beta(Z)J_Z^\star]$, proving necessity.

\end{document}